\documentclass[fontsize=11pt, a4paper]{scrartcl}
\usepackage[]{graphicx}\usepackage[]{color}
\usepackage[font=small, textfont = sl, labelfont=bf]{caption}

\makeatletter
\def\maxwidth{ %
  \ifdim\Gin@nat@width>\linewidth
    \linewidth
  \else
    \Gin@nat@width
  \fi
}
\makeatother

\definecolor{fgcolor}{rgb}{0.345, 0.345, 0.345}

\usepackage{framed}
\makeatletter
 {\par\unskip\endMakeFramed%
 \at@end@of@kframe}
\makeatother

\definecolor{shadecolor}{rgb}{.97, .97, .97}
\definecolor{messagecolor}{rgb}{0, 0, 0}
\definecolor{warningcolor}{rgb}{1, 0, 1}
\definecolor{errorcolor}{rgb}{1, 0, 0}
\newenvironment{knitrout}{}{} 

\usepackage{alltt}

\usepackage{hyperref}
\hypersetup{
    colorlinks,
    linkcolor={red!50!black},
    citecolor={blue!50!black},
    urlcolor={blue!80!black}
}

\usepackage{amsmath,amsfonts,amsthm,bm, mathtools}
\usepackage[authoryear,round]{natbib}
\usepackage{lineno}
\usepackage{dsfont}
\usepackage{xcolor}
\usepackage{geometry}
\usepackage{framed}
\definecolor{shadecolor}{rgb}{0.95, 0.95, 0.95}

\usepackage{paralist}

\newcommand{\fs}[1]{{\it\textcolor{blue}{#1}}}

\newcommand{\code}[1]{\texttt{\small{#1}}}
\newcommand{\pkg}[1]{\textsf{#1}}
\newcommand{\bvec}{\left[\begin{array}{c}}
\newcommand{\evec}{\end{array}\right]}
\newcommand{\bmat}[1]{\left[\begin{array}{*{#1}{c}}}
\newcommand{\emat}{\end{array}\right]}
\newcommand{\dimm}[2]{\underset{\phantom{.}^{#2}}{#1}}
\newcommand{\inds}[2]{_{\genfrac{}{}{0pt}{}{#1}{#2}} }
\newtheorem{proposition}{Proposition}[section]

\newcommand{\ra}{\rightarrow}
\newcommand{\Ra}{\Rightarrow}

\newcommand{\tr}{^\top}
\newcommand{\eps}{\varepsilon}

\newcommand{\real}{\mathds{R}}

\newcommand{\1}{^{-1}}

\newcommand{\iid} {\operatorname{i.i.d.}}
\newcommand{\ovec}{\operatorname{vec}}

\newcommand{\EV}{\operatorname{E}}

\newcommand{\diag} {\operatorname{diag}}

\newcommand{\snr}{\operatorname{SNR}}
\newcommand{\sd}{\operatorname{sd}}
\newcommand{\kernel}{\operatorname{ke}}

\newcommand{\rg}{\operatorname{rank}}
\newcommand{\image}{\operatorname{im}}

\newcommand{\mA}{\bm{A}}
\newcommand{\mB}{\bm{B}}
\newcommand{\mD}{\bm{D}}

\newcommand{\mH}{\bm{H}}
\newcommand{\mI}{\bm{I}}

\newcommand{\mW}{\bm{W}}
\newcommand{\mP}{\bm{P}}
\newcommand{\mV}{\bm{V}}
\newcommand{\mU}{\bm{U}}
\newcommand{\mX}{\bm{X}}
\newcommand{\mY}{\bm{Y}}
\newcommand{\mZ}{\bm{Z}}

\newcommand{\mb}{\bm{b}}

\newcommand{\mf}{\bm{f}}
\newcommand{\ms}{\bm{s}}
\newcommand{\mt}{\bm{t}}
\newcommand{\mv}{\bm{v}}
\newcommand{\mw}{\bm{w}}

\newcommand{\mSigma}{\bm{\Sigma}}
\newcommand{\btheta}{\bm{\theta}}
\newcommand{\bP}{\bm{P}}
\newcommand{\bD}{\bm{D}}
\newcommand{\bU}{\bm{U}}
\newcommand{\bV}{\bm{V}}
\newcommand{\by}{\bm{y}}

\newcommand{\bSigma}{\bm{\Sigma}}
\newcommand{\bB}{\bm{B}}
\newcommand{\bW}{\bm{W}}
\newcommand{\bX}{\bm{X}}

\newcommand{\meps}{\bm{\eps}}

\newcommand{\mGamma}{\bm{\Gamma}}
\newcommand{\mtheta}{\bm{\theta}}
\newcommand{\mTheta}{\bm{\Theta}}
\newcommand{\mPhi}{\bm{\Phi}}
\newcommand{\mXi}{\bm{\Xi}}

\newcommand{\bea}{\begin{eqnarray}}
\newcommand{\eea}{\end{eqnarray}}
\newcommand{\nn}{\nonumber}

\title{Identifiability in penalized \\ function-on-function regression models}
\author{Fabian Scheipl \& Sonja Greven \\ 
{\small LMU M{\"u}nchen}}
\IfFileExists{upquote.sty}{\usepackage{upquote}}{}
\begin{document}

\maketitle

\abstract{Regression models with functional responses and covariates constitute a powerful and increasingly important model class. However, regression with functional data poses well known and challenging problems of non-identifiability.  This non-identifiability can manifest itself in arbitrarily large errors for  coefficient surface estimates despite accurate predictions of the responses, thus invalidating substantial interpretations of the fitted models. 
We offer an accessible rephrasing of these identifiability issues in realistic applications of penalized linear function-on-function-regression and delimit the set of circumstances under which they are likely to occur in practice.
Specifically, non-identifiability that persists under smoothness assumptions on the coefficient surface can occur if the functional covariate's empirical covariance has a kernel which overlaps that of the roughness penalty of the spline estimator. 
Extensive simulation studies validate the theoretical insights, explore the extent of the problem and allow us to evaluate their practical consequences under varying assumptions about the data generating processes. A case study illustrates the practical significance of the problem. 
Based on theoretical considerations and our empirical evaluation, we provide immediately applicable diagnostics for lack of identifiability and give recommendations for avoiding estimation artifacts in practice.}

\section{Introduction}
The last two decades have seen rapid progress in regression methodology
for high-dimensional data, largely driven by applications to genomic data in the
``small $n$, large $p$'' paradigm.
In regression models with functional predictors, similar problems
arise from the fact that covariate information comes in the shape of
high-dimensional, strongly auto-correlated vectors of function evaluations.
Whenever the number of regression parameters to estimate
exceeds the number of observations, estimates are not unique and the resulting
model is not identifiable. To overcome this lack of identifiability, it
then becomes necessary to use heuristics or prior knowledge to impose additional
structural constraints like sparsity or smoothness. Results
then inherently depend -- at least to some degree -- on the assumptions
underlying the chosen regularization method. In the following, we present a
detailed analysis of the way in which smoothness assumptions combined with
properties of the data generating process affect estimation results for
function-on-function regression. We differentiate between two kinds of non-identifiability: First, \emph{simple non-identifiability} arising from low information content in the functional data. This can be cured by imposing structural assumptions like smoothness or sparsity on the estimators, i.e., by regularization of the estimators. Fig. \ref{fig:cca_6_plot} shows an example for 4 estimates under different structural assumptions all yielding identical fits in such a scenario. Second,  \emph{persistent non-identifiability} that remains despite regularization for certain combinations of models and data. Figures \ref{fig:sim:identExample} and ~\ref{fig:cca-c6-plot} show examples for the latter on synthetic and real data, respectively. 

The problem of -- especially \emph{persistent} -- non-identifiability is as yet under-appreciated in the functional data literature and analyzed here in depth for the first time. As software capable of fitting increasingly complex models with functional data becomes available
(e.g. \pkg{fda}, \citet{fda}; \pkg{fda.usc}, \citet{fda.usc}; \pkg{refund}, \citet{refund}; \pkg{PACE},
\citet{PACE}; \pkg{WFMM}, \citet{WFMM}), investigating the practical relevance of identifiability issues arising in these
models is both timely and important in this rapidly developing field. The present work
aims to perform such an investigation for the model class described in
\citet{ScheiplStaicuGreven2014} and implemented in \pkg{refund}'s \texttt{pffr}
function, while results carry over to other penalized function-on-function regression approaches such as those implemented in the \pkg{fda} package.

A popular approach in regression for functional data restricts the functional
coefficients to lie in the span of the first $K < n$ estimated eigenfunctions of
a functional covariate's covariance operator with the largest eigenvalues,
\citep[see][for example]{CardotFerratySarda1999,
CardotFerratySarda2003, YaoMuellerWang2005, Reiss:Ogden:2007, YaoMueller2010,
WuFanMueller2010}.
This functional principal component regression (FPCR) approach solves the
problem of overparameterization (i.e., non-identifiability of the functional effect)
by a -- usually drastic -- dimension reduction. The main
challenges in this approach then become 1) achieving good estimates of
the covariance's eigenfunctions (``functional principal components'' or FPCs),
eigenvalues, and FPC scores from observed functional data and 2)
choosing the regularization parameter $K$. In practice, the effect of the
functional covariate is estimated by using the first $K$ estimated FPC
scores as synthetic covariates. However, the critical assumption that the true
coefficient lies in the span of the first few empirical eigenfunctions
of a suitable (cross-)covariance operator estimate is impossible to verify empirically.
Due to the often wiggly and unsmooth nature of eigenfunctions of real data
this assumption can also lead to estimates that are difficult to interpret or
implausible to practitioners.

An alternative approach is to make assumptions on the functional coefficients
informed by insights into the problem at hand, e.g., sparsity or smoothness of
functional coefficients, and to estimate these functional coefficients subject
to an appropriate penalty (e.g., LASSO or smoothness penalties).
In this work, we will focus on smooth spline-based penalized regression models
for functional responses with functional covariates as described in
\citet{Ivanescu2011} and \citet{ScheiplStaicuGreven2014},
which constitute a powerful and flexible model class able to deal with the wealth of functional data increasingly collected in many fields of science.
Nevertheless, our considerations carry over to other approaches to estimate
smooth coefficient functions, such as approaches using derivative-based
penalties as advocated by \citet{RamsaySilverman2005}.
This paper describes the data settings in which penalized models remain unidentifiable
despite the penalty in Section \ref{sec:issue} and develops and evaluates suitable diagnostics and modified penalties for such settings.

Identifiability issues in functional regression have previously been discussed
in \citet{CardotFerratySarda2003} in the context of functional regression models
for scalar responses and also, briefly, in the context of models with both
functional responses and functional predictors by \citet{HeMullerWang2003},
\citet{ChiouMullerWang2004} and \citet{PrchalSarda2007}.
While results therein provide conditions for the theoretical existence and
unicity of solutions based on functional analysis arguments, they do not yield
empirically verifiable criteria to determine whether the conditions for unicity
are violated for a given data set. They also always assume that the true
coefficient surface lies in the space spanned by the eigenfunctions of a
(cross-)covariance operator.
As far as we are aware, case studies in the previous literature have implicitly
assumed that this assumption and the necessary conditions based on it will be
satisfied for observed data. This is problematic since
1) the theoretical conditions found in the previous literature are impossible
to satisfy, or at least verify, on finite samples of functional data in
finite resolution, and
2) our experience with applications of functional regression models as well as
results from simulation studies indicate that persistent non-identifiability leading to spurious coefficient estimates may occur quite regularly. This is obviously a
concern for applied statisticians desiring \emph{interpretable} regression
models associating functional covariates and (functional) responses.

Instead of relying on the functional analysis arguments suitable for
investigations of asymptotic properties of the theoretical model, we use simple
linear algebra to derive a condition for unicity of coefficient surface
estimates in realistic, finite sample data settings in which functional
covariates are observed with finite resolution in Section \ref{sec:issue}.
This allows us to give a
necessary and sufficient condition for persistent non-identifiability in penalized
function-on-function regression that is empirically verifiable and thus
applicable in realistic problems. The criterion is based on the amount of overlap between the kernel of the penalty matrix and the kernel of the design matrix for the functional effect. Simulation studies indicate that, in practice, severe errors due to non-identifiability are strongly associated with our criterion; thus, this criterion is the first one
that can be used in a wide variety of applications to assess identifiability or
lack thereof.
Our analyses also indicate that many widely used preprocessing techniques for
functional data which replace observed curves with spline-based or FPC-based
low-rank approximations  (i.e., pre-smoothing) or the centering of individual
observed curves will considerably increase the likelihood of identifiability
issues in many settings.

\begin{figure}[!htbp] {\centering}
\includegraphics[width=\textwidth]{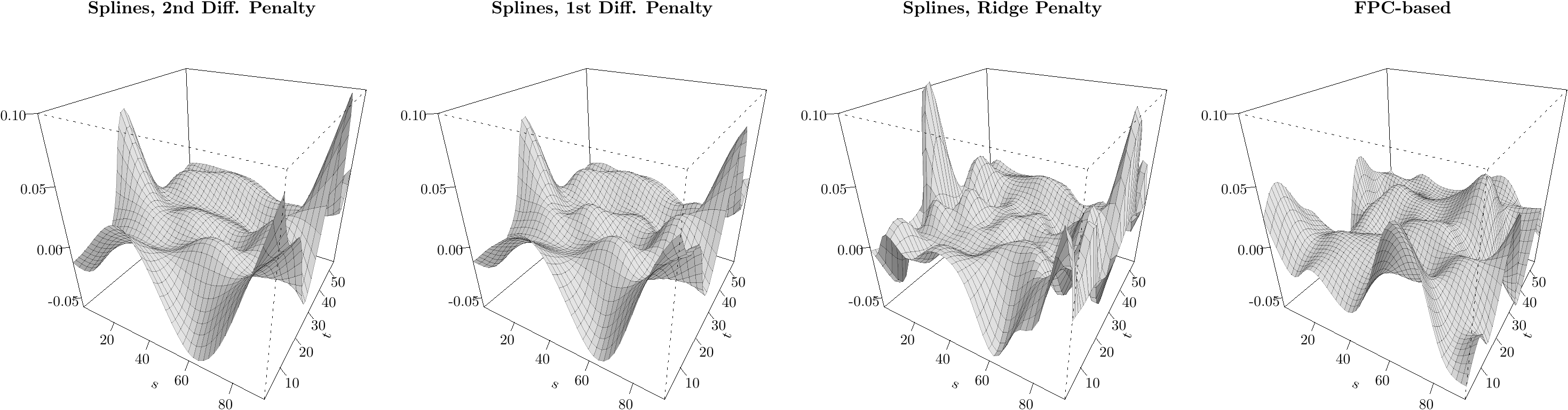}
\caption{\label{fig:cca_6_plot} Four coefficient surfaces for regressing RCST-FA on
pre-smoothed CCA-FA truncated to its first 6 empirical FPCs. All lead to identical fitted values. First [second] panel: Tensor product spline-based fits (12 marginal cubic B-spline basis functions) with second [first] order difference penalty. Third panel: Tensor product spline-based fit with a Ridge penalty. Fourth panel: FPC regression result (6 FPCs). All fits performed with \code{pffr()}.}
\end{figure}

We use the well-known \code{dti} data set -- publicly available in R-package
\pkg{refund} \citep{refund} -- to illustrate the issues we discuss here.
These data contain spatially indexed, i.e., \emph{functional}, measurements of fractional anisotropy (FA), a proxy variable for neuronal health, along 3 cerebral white matter tracts (WMTs)
of multiple sclerosis (MS) patients. To illustrate how strongly different
structural assumptions about the regression coefficient surface can affect the
results, we fit simple univariate functional linear models
\mbox{$E(Y(t)) =  \beta_0(t) + \int X(s)\beta(s,t)ds$} regressing FA along the right cortico-spinal WMT (RCST-FA, $Y(t)$) on pre-smoothed FA along the \emph{corpus callosum} WMT (CCA-FA, $X(s)$).
Figure \ref{fig:cca_6_plot} shows estimated coefficient surfaces $\hat\beta(s,t)$
for penalized spline based fits with second or first order difference penalties (two leftmost panels) or a ridge penalty (third panel from left), as well as the coefficient surface implied by a functional principal component (FPC) based fit (right). Even though the surfaces are quite different, they result in (practically) identical fitted values in this example for a setting with simple non-identifiability. The first and third panels from the left in the top row of Figure~\ref{fig:cca-c6-data-fits} show the data used in this example.
Graphical examples of the second kind of non-identifiability that persists even under penalization are shown in Figures~\ref{fig:sim:identExample} and \ref{fig:cca-c6-plot} for synthetic and real data, respectively. 

Our paper is structured as follows: Section~\ref{sec:model} defines the model and data structure under discussion. We present an accessible rephrasing of the fundamental issue (Section~\ref{sec:issue}) and derive necessary and sufficient conditions for settings in which $\beta(s,t)$ is or is not identifiable in Section~\ref{sec:issue:eqrep}. Section~\ref{sec:sim} follows up with an analysis of the scope of the problem based on simulated data, while Section~\ref{sec:example} describes a real-world example of the issue. The main conclusions we draw from our analysis are that the complexity of observed functional covariates puts hard limits on the identifiability of coefficient surfaces in a number of ways, that these limits can be diagnosed based on the data at hand, and that pre-processing of functional covariates can often exacerbate identifiability issues.

\section{Model and Data Structure}\label{sec:model}
In the following, bold symbols denote vectors and matrices, and calligraphic letters denote function domains, function spaces, or sets of functions or
vectors.

Define a simple function-on-function regression model as
\begin{align}
Y_i(t) &= \int_{\mathcal{S}} X_{i}(s)\beta(s,t)ds + \eps_{it}, 
\label{eq:model}
\end{align}
where $Y_i(t)$ and $X_i(s)$, $i=1,\dots,N$, are functional responses and
covariates on closed intervals $\mathcal{T}$ and $\mathcal{S}$ in $\mathbb{R}$,
respectively, and assume that they are realizations of zero-mean square 
integrable stochastic processes  $Y(t) \in L^2[\mathcal{T}]$ and  
$X(s) \in L^2[\mathcal S]$ with continuous covariance functions, respectively. 
To simplify notation and exposition but without loss of generality, we assume 
that $E(Y_i(t))\equiv 0$ and $E(X_i(s))\equiv 0$. Propositions \ref{myfirstproposition} to
\ref{unique} in Section \ref{sec:issue} directly carry over to any model with an additive predictor that includes terms like $\int_{\mathcal{S}} X_{i}(s)\beta(s,t)ds$.
The further development in Section \ref{sec:issue} leading up to Proposition \ref{invert} assumes that the estimate for $\beta(s,t)$ minimizes a (penalized) quadratic loss function, which is equivalent to maximizing the likelihood of a model with i.i.d. Gaussian errors $\eps_{it}$, but does not strictly speaking depend on distributional assumptions about $\eps_{it}$ and is also very similar to the system of equations solved in each iteration of the penalized iteratively re-weighted least squares algorithm \citep[P-IWLS;][]{Wood2000} used to fit additive models like \eqref{eq:model} for non-Gaussian responses.

Due to the assumptions on the functional covariates, they can be represented by
a Karhunen-Lo\`{e}ve expansion 
\begin{align}\label{karhunen}
X_i(s) = \sum^\infty_{m=1} \xi_{im} \phi_m(s), 
\end{align}
with orthonormal $\phi_m(s)$, 
$\int_{\mathcal{S}} \phi_m(s)\phi_{m'}(s)ds =\delta_{mm'}$, and uncorrelated
zero-mean FPC scores $\xi_{im}$ with variances $\nu_1 \geq \nu_2 \geq \dots \geq 0, m \in  \mathbb{N}$.
The $\nu_m$ and $\phi_m(s), m=1,\dots,M,$ are the ordered eigenvalues and eigenfunctions of the covariance operator $K^X$ of $X(s)$, respectively, with the covariance function given by Mercer's theorem as 
\begin{align}
k^X(s,s')=\EV\left(X(s)X(s')\right) = \sum_{m=1}^\infty \nu_m \phi_m(s) \phi_m(s').
\end{align}

Since estimating $\beta(s,t)$ is an inverse problem, some kind of regularization
is required. Functional principal component based approaches like, for example,
\citet{YaoMuellerWang2005} restrict $\beta(s, t)$ to lie in the span of
the first $K$ estimated eigenfunctions $\hat\phi_m(s),\; m=1,\dots,K$ for all $t$.
The number of eigenfunctions $K$ that is used serves as the (discrete)
regularization parameter. In contrast, we will discuss and
analyze a penalized approach. 
The underlying assumption 
is that $\beta(s,t)$ is a smooth function that can be well represented as a linear combination
of suitable basis functions defined on $\mathcal{S} \times \mathcal{T}$.  

In practice, functional responses $Y_i(t)$ and functional covariates $X_i(s)$
are observed on grid points \mbox{$\ms_i=(s_{i1},\dots,s_{i{S_i}})$} and
\mbox{$\mt_i=(t_{i1},\dots,t_{i{T_i}})$}. 
For simplicity,  we assume those to be
identical vectors $\ms, \mt$ 
with lengths $S$ and $T$, respectively, for each observation $i$. In the following, expressions $a(\ms)$ or $a(\mt)$ with a bold argument
denote the vector of evaluations of $a(\cdot)$ on the respective grid, e.g.,
$a(\ms) = (a(s_1), \dots, a(s_S))\tr$.

Model \eqref{eq:model} can then be approximated for observed data as  
\begin{align}
\dimm{Y_i(\mt)\tr}{1 \times T} &\approx
       \left(\dimm{\mw{\phantom{(}}}{S \times 1} \cdot \dimm{X_{i}(\ms)}{S \times 1}\right)\tr  \dimm{\beta(\ms, \mt)}{S \times T}
      + \dimm{\bm{\eps}_{i}}{1 \times T}, \label{eq:scalarobsmodel}
\intertext{
\noindent with $\beta(\ms,\mt)=\left[\beta(s_j,t_k)\right]\inds{j=1,\dots,S}{k=1,\dots,T}$ and 
$\bm{\eps}_{i}=(\eps_{it_1}, \dots, \eps_{it_T})\tr$. 
We also define a weight vector $\mw$ for numerical integration, e.g.~$\mw=(w_j)_{j=1,\dots,S}$ for simple quadrature via Riemann sums, with $w_j$ the length of the sub-interval of $\mathcal{S}$ represented by $s_j$. 
The symbol $\cdot$ denotes element-wise multiplication. 
The coefficient surface $\beta(\ms, \mt)$ is represented using a tensor product spline basis}
\beta(\ms, \mt) &= \dimm{\mB_s}{S \times K_s} \dimm{\mTheta_{\phantom{.}}}{K_s \times K_t} \dimm{\mB_t}{K_t \times T}\tr, \label{eq:betatensorrep}
\end{align}
with basis matrices $\mB_s$ and $\mB_t$ of $K_s$ and $K_t$ basis functions evaluated in $\ms$ and $\mt$, respectively, 
and spline coefficient matrix $\mTheta$. The roughness penalty matrix for the surface is given by 
$\mP \equiv \mP(\lambda_s, \lambda_t) = \lambda_s (\mP_s \otimes \mI_{K_t}) + \lambda_t (\mI_{K_s} \otimes \mP_t)$ \citep{Wood04lowrank}, where 
$\lambda_s, \lambda_t$ are smoothing parameters to be estimated from the data
and $\mP_s$ and $\mP_t$ are the fixed and known marginal penalty matrices for 
the $s$- and $t$-directions, respectively. 

Estimation and inference  is described in more detail in \citet{Ivanescu2011} and \citet{ScheiplStaicuGreven2014}. In the following, our considerations are not limited to simple models such as model \eqref{eq:model}, but carry over to more general models $\tilde Y_i(t) =  \eta_i(t) + \int_{\mathcal{S}} X_{i}(s)\beta(s,t)ds + \eps_{it}$ by using 
$Y_i(t) = \tilde Y_i(t) -  \eta_i(t)$. The additive predictor $\eta_i(t)$ represents the sum of other terms in the model such as a global functional intercept $\beta_0(t)$,
index-varying linear or smooth effects of scalar covariates $x$ like $x_i \beta(t)$ or $f(x_i, t)$, scalar or functional random effects, etc. \cite{ScheiplStaicuGreven2014} contains methods and applied examples for this class of flexible additive functional regression models. Of course, these more general models may suffer from additional identifiability problems caused by collinearity or concurvity of the terms in the additive predictor which are outside the scope of this paper.
While we focus our discussion on a spline-based approach, the representation in \eqref{eq:betatensorrep} also accommodates other choices of basis functions and penalties.

\section{Identifiability}\label{sec:issue}
In this section, we discuss potential sources of non-identifiability in model
\eqref{eq:model}. The first subsection restates some known results on these
issues for the theoretical model \eqref{eq:model} with truly functional
observations $X_i(s)$ and $Y_i(t)$. Subsection \ref{sec:issue:eqrep} then
discusses identifiability for the finite resolution vector data available
in practice.

\subsection{Identifiability in the theoretical model}\label{sec:issue:theory}

It is well known (e.g. \citet[][c.f.~p.~5] {PrchalSarda2007},
\citet[][Th.~4.3.~c]{HeMullerWang2003}) that coefficient surfaces are identifiable
only up to the addition of functions in the kernel of $K^X$, i.e., if
$\beta(\cdot,t)$ fulfills model \eqref{eq:model}, so does
$\beta(\cdot,t) + \beta_K(\cdot,t)$ for any $\beta_K(\cdot,t)$ with
$\int_\mathcal{S} k^X(s,v) \beta_K(v, t) dv = 0$ for all $s, t$. Thus, we have
identifiability only when the kernel is trivial.
\begin{proposition}
\label{myfirstproposition}
The coefficient surface $\beta(s,t)$ in \eqref{eq:model}
is identifiable if and only if
\linebreak
$\kernel(K^X) = \{0\}$.
\end{proposition}
An important secondary consequence is that predicted responses
\mbox{$\int_{\mathcal{S}} X_{i}(s)\beta(s,t)ds$} can be entirely unaffected by
large changes in $\beta(s,t)$. Thus, strategies for detection of identifiability
problems cannot be based on predictive performance in cross-validation,
bootstrapping and related methods.

Non-identifiability in Proposition \ref{myfirstproposition} occurs when
$\kernel(K^X)$ is non-trivial, when the eigenfunctions in \eqref{karhunen}
with non-zero eigenvalues $\nu_m$ do not span the $L^2[\mathcal{S}]$.
While it is possible to assume a trivial kernel in theory
\citep[e.g.][equation (4)]{PrchalSarda2007},
in practice, functional covariates are observed on a finite number of
grid points $S$ and the empirical covariance for $N$ observations thus can have
at most $\min(N, S)$ non-zero eigenvalues. As is exploited in functional
principal component analysis \citep[e.g.][]{RamsaySilverman2005}, functional
observations are often simple enough to be represented accurately by a
relatively small number of eigenfunctions, with eigenvalues of higher order
small compared to noise or measurement error. It is also wide-spread practice to
use pre-smoothed versions of observed functional covariates as inputs for models
like \eqref{eq:model} \citep[e.g.~][]{James2002, RamsaySilverman2005}, and these
will have a non-empty kernel since they are represented as linear combinations
of a limited number of basis functions.
Basis function representations of $X(s)$ are also used when sparsely or incompletely
observed functional covariates have to be imputed on a grid of $s$-values to be used as
inputs for model \eqref{eq:model} \citep[e.g.][]{Goldsmith:etal:2011}.

In the following section, we thus investigate identifiability problems for
finite-sample finite-resolution functional data and the interplay between
the rank of the observed covariance, the rank of the basis used to represent
$\beta(s,t)$ in $s$-direction and the penalty used in the penalized estimation
approach for $\beta(s,t)$ introduced in Section \ref{sec:model}.

\subsection{Identifiability in practice}\label{sec:issue:eqrep}

\subsubsection*{Rank-deficient design matrix}

In practice, $\beta(s,t)$ is represented as a linear combination of a finite number $K_s K_t$ of basis functions, see \eqref{eq:betatensorrep}. For the following, we will assume that the corresponding approximation error is negligible and that a suitably flexible basis has been chosen for $\beta(s,t)$. 
Combining \eqref{eq:scalarobsmodel} and \eqref{eq:betatensorrep},
we can write the model as
\begin{equation}
\dimm{\mY_{{\phantom{i}}}}{N \times T\phantom{i}} =
\dimm{\mX_{{\phantom{i}}}}{N \times S\phantom{i}}
\dimm{\mW_{\phantom{i}}}{S \times S\phantom{i}} \dimm{\mB_s}{S \times K_s}
\dimm{\mTheta_{\phantom{i}}}{K_s \times K_t\phantom{i}} \dimm{\mB_t^T}{K_t \times T\phantom{i}} +
\dimm{\meps_{\phantom{i}}}{N \times T\phantom{i}},
\end{equation}
where $\mY = \left[Y_i(t_j)\right]\inds{i=1, \dots, N}{j=1, \dots, T}$,
$\mX=\left[X_i(s_l)\right]\inds{i=1,\dots,N}{l=1,\dots,S}$,
$\bm W = \diag(\mw)$ and
$\meps = \left[\eps_{it_j}\right]\inds{i=1, \dots, N}{j=1,\dots, T}$.
Using $\ovec(\bm{ACB}) = (\mB^T \otimes \bm A) \ovec(\bm{C})$ yields
\begin{equation} \label{idmodel}
\dimm{\ovec(\mY)_{{\phantom{i}}}}{NT \times 1} = [\dimm{\mB_t}{T \times K_t}
\otimes (\dimm{\mX_{{\phantom{i}}}}{N \times S\phantom{i}}
\dimm{\mW_{{\phantom{i}}}}{S \times S} \dimm{\mB_s}{S \times K_s})]
\dimm{\ovec(\mTheta)_{{\phantom{i}}}}{K_t K_s \times 1} + \dimm{\ovec(\meps)_{{\phantom{i}}}}{NT \times 1}.
\end{equation}
In the linear regression model \eqref{idmodel} for $\bm y =\ovec(\mY)$, the parameter vector $\mtheta = \ovec(\mTheta)$ is identifiable if and only if the design matrix $\mD = \mB_t \otimes (\mX \mW \mB_s)$ is of full column-rank.

The rank of $\mD$ is equal to $\rg(\mD) = \rg(\mB_t)\rg(\mX \mW \mB_s)$.
$\mB_t$ will typically be of full rank $K_t$ as long as $K_t \leq T$, as the
$K_t$ spline functions form a basis and the columns of $\mB_t$ are thus linearly
independent for non-pathological cases.
For $\mX$, let $\mX = \mXi \mPhi$ be the empirical version of the
Karhunen-Lo\`{e}ve expansion \eqref{karhunen}, where
$\mX\tr\mX = \mPhi\tr \bm\Lambda \mPhi$ with $\mPhi\tr$ an $S \times M$ orthonormal
matrix of eigenvectors, $M = \rg(\mX)$,
$\bm\Lambda = \operatorname{diag}(\hat\lambda_1, \dots, \hat\lambda_M)$ a
diagonal matrix of ordered positive eigenvalues and $\mXi$ containing $M$
columns of estimated scores with empirical variances
$\hat\lambda_1, \dots, \hat\lambda_M$.
Then, by construction, the matrix $\mX \mW \mB_s$ is at most of rank
$\min(N, M, K_s, S) = \min(M, K_s)$, since $M \leq \min(N, S)$.

We then have the following proposition:
\begin{proposition}
\label{sonjaproposition1}
Assume that $\mB_t$ is of full rank $K_t$. Then, the design matrix $\mD = \mB_t \otimes (\mX \mW
\mB_s)$ in model \eqref{idmodel} is
rank-deficient if and only if
\begin{enumerate}
 \item[a)] $M < K_s$ or
 \item[b)] if $M \geq K_s$, but $\rg(\mPhi \mW \mB_s) < K_s$ .
\end{enumerate}
\end{proposition}
\begin{proof} See Appendix \ref{sec:proof-sonjaproposition1}.
\end{proof}

Proposition \ref{sonjaproposition1} yields a direct criterion to check for
rank-deficient design matrices. Case a) corresponds to a low-rank covariance for
the $X$-process. In this case, the functional predictor does not carry enough
information, as measured by the number of eigenfunctions $\phi_m(s)$ with non-zero eigenvalues, compared to the number of parameters to estimate.
Case b) means that even if $M \geq K_s$, non-identifiability can occur if
the span of the basis used for $\beta(s,t)$ in $s$-direction contains functions
in $\kernel(K^X)$, as measured by numerical integration using the integration
weights $\mw$. More intuitively, this means that the basis for the
$s$-direction of $\beta(s,t)$ accommodates modes of variation orthogonal to those of the $X(s)$-process.

Note that identifiability is determined by the interplay between the complexity of the $X(s)$ and the coefficient basis for $\beta(s,t)$.
Thus, more data will not necessarily resolve identifiability issues:
A finer grid $\ms$ will only eliminate identifiability problems present on a coarser grid if there is sufficient small-scale structure in $X(s)$ that is
also present in the basis used for $\beta(s, t)$ in $s$-direction.
Increasing the sample size $N$ will likewise eliminate problems with
identifiability only if the low-rank of the covariate's covariance is due to
small sample size.

\subsubsection*{The effect of the penalty}

In cases of non-identifiability, the best we can hope for is partial
identifiability of the parameters in a parameter subspace, i.e., identifiability
under additional assumptions on the parameters. In this vein, functional
principal component regression \citep[e.g.][]{YaoMuellerWang2005} restricts
$\beta(s, t)$ to lie in the span of the first $K$ eigenfunctions
$\phi_m(s),\; m=1,\dots,K,$ for all $t$.
Remaining problems then include the fact that the $\hat\phi_m(s)$ are estimated
quantities in practice, with corresponding measurement error, and the choice of
$K$, which can strongly affect the shape of the resulting function estimate
\citep{Crainiceanu:etal:2009}. Also, this approach couples assumptions on the
shape of $\beta(s, t)$ to properties of the space spanned by the retained $\phi_m(s)$. In particular, their smoothness determines that of $\beta(s,t)$, and inclusion of higher-order eigenfunctions often leads to wiggly surface estimates that are hard to interpret and unstable under replication.

Here, we focus on a penalized approach assuming smoothness of $\beta(s,t)$ and investigate the effects of the penalty on identifiability. While the use of a penalized approach is well-known to avoid identifiability problems due to high correlation between observations at neighboring grid points \citep[e.g.][Ch. 15.2]{RamsaySilverman2005}, the full
interplay between penalty and identifiability is, we believe, not fully
understood and under-appreciated.

Consider again the design matrix $\mD = \mB_t \otimes (\mX \mW \mB_s)$ of rank
$d$ with the singular value decomposition
$\mD = \mV \mSigma \mU\tr = (\mV_t \otimes \mV_s) (\mSigma_t \otimes \mSigma_s)(\mU_t\tr \otimes \mU_s\tr)$,
with $\mV_t \mSigma_t \mU_t\tr$ and $\mV_s \mSigma_s \mU_s\tr$ the singular value
decompositions of $\mB_t$ and $\bm D_s := \mX \mW \mB_s$, respectively.
Let indices $_+$ and $_0$ denote the corresponding sub-matrices obtained by
removing columns and/or rows corresponding to zero and non-zero singular values,
respectively. We assume in the following that $\mB_t$ is of full rank $K_t$.
Then,
$\mD = \mV_+ \mSigma_+ \mU_+\tr = (\mV_t \otimes \mV_{s+}) (\mSigma_t \otimes \mSigma_{s+})(\mU_t\tr \otimes \mU_{s+}\tr)$.

Thus, for any given $\mtheta_\star$ with
$\mD\mtheta_\star =: \mf$, there exists a linear subspace
$\mathcal{H}_f \subset \real^{K_sK_t}$ of dimension $(K_sK_t -d)$ given by
\mbox{$\mathcal{H}_f = \{ \mtheta \in \mathds{R}^{K_sK_t} : \mD \mtheta = \mf\} = \{\mtheta_\star + \mtheta_0: \mtheta_0 \in \image(\mU_0)\}$}
that yields identical fits $\mf$.

If we assume our parameter function to come from a space of smooth functions, we can select the smoothest solution on a given hyperplane $\mathcal{H}_f$ by minimizing $\btheta\tr \bP \btheta$ for a suitable penalty matrix $\bP$ that penalizes roughness of the function parameterized by $\btheta$. \emph{Simple} non-identifiability occurs if this minimum is unique, \emph{persistent} non-identifiability occurs if it is not.
We have the following proposition regarding uniqueness of the corresponding minimum.
\begin{proposition} \label{unique}
 Let
 $\bP = \lambda_s (\bm{I}_{K_t} \otimes \bP_s) + \lambda_t (\bP_t \otimes \bm{I}_{K_s})$,
 with $\bP_s$ and $\bP_t$ positive semi-definite matrices.
 Assume that $\bm B_t$ is of full rank $K_t$, that
 $\rg(\mP_t) < K_t$ and that $\lambda_s >0, \lambda_t \geq 0$.
 Then, for any $\mf \in \image(\bm D)$ there is a unique minimum
$\min_{ \btheta \in \mathcal{H}_f }\{ \btheta\tr \bP \btheta\}$ if and only if $\kernel(\bm D_s\tr  \bm D_s ) \cap \kernel(\bP_s) = \{\bm 0\}$.
\end{proposition}
\begin{proof} See Appendix \ref{sec:proof-unique}.
\end{proof}

The assumption that $\mP_t$ is of less-than-full rank is natural in the context of
derivative-based penalties and excludes cases like the ridge penalty
$\mP_t =  \mI_{K_t}$, which would have the same effect as a full-rank penalty
$\mP_s = \mI_{K_s}$, even in cases of a kernel overlap between $\mD_s\tr\mD_s$
and $\bP_s$. A potentially full-rank $\mP_t$ would change the 'if and only if' in
Propositions \ref{unique} and \ref{invert} below to 'if'.

Proposition \ref{unique} shows that in the case of a kernel overlap
$\kernel(\mD_s\tr\mD_s ) \cap \kernel(\bP_s) \neq \{\bm 0\}$, the additional
side condition $\btheta\tr \bP \btheta \ra \min$ does not yield a unique
smoothest point on the hyperplane and the model remains unidentifiable.
On the other hand, if there is no kernel overlap, there is a unique smoothest
point $\btheta_f \in \mathcal{H}_f$ 
and this unique point has the form of a projection of $\btheta$ 
along the hyperplane. Note that for the ridge penalty $\bP = \lambda \mI_{K_sK_t}$, one
obtains the projection onto the image $\image(\bD)$, which sets the part in the
kernel of $\bD\tr\bD$ to zero.
More generally, a smoothness penalty $\bP$ generates a projection that may have
a non-zero component in the kernel of $\bD\tr\bD$ if this yields a smaller
overall penalty value.
In this case of no kernel overlap, we thus have a weak form of identifiability,
which guarantees that there is a unique smoothest representative on any
hyperplane of parameters giving the same conditional distribution for $\mY$.



This characterization, which only requires checking of design matrix and penalty
in $s$-direction and not for the full model, carries over to the penalized
maximum likelihood or least squares estimation problem
\bea\label{penest}
\min_{\btheta} \{\|\by - \bD \btheta\|^2 + \lambda_s \btheta\tr(\bm{I}_{K_t} \otimes \bP_s) \btheta + \lambda_t \btheta\tr(\mP_t \otimes \mI_{K_s}) \btheta\}
\eea for some $\lambda_s> 0, \lambda_t \geq 0$.


\begin{proposition} \label{invert}
Assume that $\bm B_t$ is of full rank $K_t$, that $\rg(\mP_t) < K_t$ and that $\lambda_s >0, \lambda_t \geq 0$. Then, there is a unique penalized least squares solution
for \eqref{penest} if and only if $\kernel(\mD_s\tr
\mD_s ) \cap \kernel(\bP_s) = \{\bm 0\}$.
\end{proposition}

\begin{proof} See Appendix \ref{sec:proof-invert}.
\end{proof}
Proposition \ref{invert} gives a criterion for the uniqueness of the penalized
least squares solution. We can show how the penalty achieves this uniqueness by
writing \eqref{penest} as a nested minimization problem. Here, the outer
minimization finds the $\mf = \bD\btheta$ with optimal fit to the data, and the
inner minimization minimizes the penalty term over $\mathcal H_f$ to obtain the
smoothest solution for a given level of residual variation.
\begin{align*}
&\min_{\btheta} \{\|\by - \bD \btheta\|^2 + \lambda_s \btheta\tr (\bm{I}_{K_t} \otimes \bP_s) \btheta + \lambda_t \btheta\tr (\bP_t \otimes \bm{I}_{K_s}) \btheta \} \\
=&\min_{\mf \in \image(\bD)} \min_{\btheta} \{\|\by - \bm f\|^2 + \btheta\tr \mP \btheta \qquad \text{s.t.} \quad \bD\btheta = \mf\} \\
=&\min_{\mf\in \image(\bD)} \{\|\by - \bD \btheta_{f}\|^2 + \btheta_{f}\tr \mP \btheta_{f} \} \\
=&\min_{\mv_+\in \real^d} \{\|\by - \bV_+ \bSigma_+ \mv_+\|^2 +
\mv_+\tr\mU_+\tr \mH\tr \mP \mH \mU_+ \mv_+
\end{align*}
where $\btheta_f = \mH\mU_+ \mv_+$ for a given $\mf$ is uniquely defined as in \eqref{thetac}, with
$\bm v_+ = \bm\Sigma_+\1\bm V_+\tr \bm f$ and
$\mH = (\mI_{K_sK_t} - \bU_0 (\bU_0\tr \bP \bU_0)^{-1} \bU_0\tr\bP )$,
 if $\kernel(\mD_s\tr
\mD_s ) \cap \kernel(\bP_s) = \{\bm 0\}$.

As $ \mV_+ \bm \bSigma_+$ is a matrix of full column rank $d$, 
this minimization problem has a unique solution that, for given $\lambda_s, \lambda_t$, balances the fit to the data and smoothness.

To summarize, in the case of no kernel overlap, i.e., $\kernel(\mD_s\tr
\mD_s ) \cap \kernel(\bP_s)= \{\bm 0\}$, we obtain a weaker form of identifiability
even when the design matrix $\bm D$ is not of full rank,
which guarantees that there is a unique smoothest representative on any
hyperplane of parameters giving the same conditional distribution for $\bm y$.
Then, there will also be a unique solution to the penalized estimation problem,
which is the smoothest representative on the hyperplane of possible solutions
with equally good fit.


In practice, $\lambda_s$ and $\lambda_t$ are estimated from the data.
We here do not investigate the more complex case when $\lambda_s, \lambda_t$ are
not fixed. It should also be noted that
$(\bD\tr \bD + \lambda_s (\bm{I}_{K_t} \otimes \bP_s) + \lambda_t (\bP_t \otimes \bm{I}_{K_s}))$
may still be close to singular even in cases of no kernel overlap if smoothing
parameters are very small, with corresponding reduced stability in estimation.

\subsection{Diagnostics, practical recommendations and countermeasures}\label{sec:issue:diag}

In order to safeguard against misleading coefficient estimates in practical
applications of functional regression, it is necessary 1) to develop empirical
criteria for diagnosing problematic data settings in which the coefficient
function is not identifiable based on the available data, or where only  the
penalty ensures unique estimates, 2) to avoid pre-processing protocols and tuning parameters that
increase the likelihood of identifiability issues and
3) to develop improved algorithms for estimating function-on-function effects that are
less prone to severely misleading estimates in problematic settings.

\paragraph{Diagnostics}

We are interested in identifying settings with \emph{simple} non-identifiability in which only the penalty term guarantees the existence of a unique solution. Following
Proposition \ref{sonjaproposition1}, the most direct approach to do so is to
compute the condition number of $\mD_s\tr \mD_s=\left(\mX\mW\mB_s\right)^T\mX\mW\mB_s$ and choose a suitable  cut-off ($10^6$, in the following) for numeric rank deficiency.

In addition, propositions \ref{unique} and \ref{invert} indicate that a measure of the degree of overlap between the spans of  $\kernel(\mD_s\tr \mD_s )$ and  $\kernel(\bP_s)$ can be used to detect \emph{persistent} non-identifiability that remains despite the penalization. In our empirical evaluation of such measures, we found that a measure for the distance between the spans of two matrices introduced in \citet{LarssonVillani2001}, when modified  for our setting, showed the most promise as the resulting measure is free of tuning parameters and can be computed quickly from the data.

In particular, we  modify the original definition of Larsson-Villani in order to
 accommodate two matrices of unequal column numbers. We then define
the amount of overlap 
$\bigcap_{LV}$  between the span of  two matrices
$\mA \in \real^{n \times p_A}$, $\mB \in \real^{n \times p_B}$, $n > p_A, p_B$, by
\begin{align*}
\bigcap\nolimits_{LV}(\mA, \mB) =  \operatorname{trace}(\mV_B\tr \mV_A \mV_A\tr \mV_B).
\end{align*}
Here, $\mV_Z$ is a matrix containing the left singular vectors of the matrix
$\mZ$ and is thus an orthogonal matrix spanning the same column space as $\mZ$, 
$\mZ \in \{\mA, \mB\}$.
It is easy to see that this measure 
 is symmetric, $\bigcap_{LV}(\mA,\mB) = \bigcap_{LV}(\mB,\mA)$.
Similarly to Theorem 2 in \citet{LarssonVillani2001}, one can also show that $\bigcap_{LV}(\mA,\mB) \in [0,\min(p_A, p_B)]$,
with the overlap assuming its maximum of $\min(p_A, p_B)$
 iff $\mA \in\image(\mB)$ or
$\mB \in \image(\mA)$ and its minimum of~0 iff $\mA \in \image(\mB_\bot)$ 
or
$\mB_\bot \in \image(\mA)$, 
where the $n \times (n-p_Z)$ matrix ${\mZ}_\bot$ is the orthonormal  complement of $\mZ$, $\mZ \in \{\mA, \mB\}$.


To measure the degree of overlap between the kernels of $\bD_s\tr\bD_s$ and of $\bP_s$, we could use $\bigcap_{LV}((\bD_s\tr\bD_s)_\bot, \bP_{s\bot})$,
as the span of $(\bD_s\tr\bD_s)_\bot$ corresponds to the kernel of $\bD_s\tr\bD_s$ and the span of $\bP_{s\bot}$ corresponds to the kernel of $\bP_s$.
In the following, we will however use
\begin{align}
\bigcap\nolimits_{X_\bot \mathcal{P}_\bot} = \bigcap\nolimits_{LV}((\mX^T\mX)_\bot, \mW \mB_s \mP_{s\bot})
= \bigcap\nolimits_{LV}((\mX\tr)_\bot, \mW \mB_s \mP_{s\bot}),
\label{eq:dxp}
\end{align}
as this choice obtained slightly better sensitivity and specificity for the detection of problematic settings in our simulations in Section \ref{sec:sim}. Moreover, one can easily show that
\begin{align*}
\kernel(\bD_s\tr\bD_s) \cap \kernel(\bP_s) =  \{\bm{0}\}
 \quad \Leftrightarrow \quad
\kernel(\bX\tr \bX) \cap \{\bW\bB_s \bm{x} \ |\  \bm{x} \in \kernel(\bP_s)\}=  \{\bm{0}\}
\end{align*}
such that the two formulations address the same question.
For $\bW = \bm{I}_S$, this measure has the interpretation of   the overlap  
between the empirical null-space of the
observed $X(s)$ process or $\kernel(K^X)$ and $\mathcal{P}_{s\bot}$, the space of functions not penalized by the penalty defined by $\bP_s$,  evaluated on the grid given by $\ms$.
It can be determined quickly and accurately  before the model is fit. 
Problematic cases are indicated by overlap measures $\geq 1$, as this is indicative of an at least one-dimensional sub-space of functions contained in
the kernel overlap.


\paragraph{Practical recommendations}
The theoretical results suggest several recommendations for pre-processing of functional covariates and choice of the penalty in practice.

\begin{enumerate}

\item Pre-smoothing of functional covariates is commonly done to remove measurement error and/or obtain functions on a common grid \citep[e.g.~][]{James2002, RamsaySilverman2005,Goldsmith:etal:2011}. If the resulting (effective) rank of the smoothed covariate process drops below $K_s$, this will lead to models that are only identifiable through the penalty term. We thus recommend to use a sufficiently large number of FPCs and/or spline basis functions if such pre-processing is required.

As a peculiar consequence of this point  it may be preferable in some cases to accept a small amount of measurement error-induced attenuation in $\hat\beta(s,t)$ based on noisy, unprocessed $X(s)$ in order to avoid a potentially much larger non-identifiability-induced error in $\hat\beta(s,t)$ based on low-rank, pre-processed $X(s)$.

\item
Curve-wise centering of  functional covariates such that $\sum^S_{l=1} X_i(s_l)=0$ for all $i$ is sometimes used e.g.~in the context of spectroscopy data to remove the optical offset \citep[c.f.][]{Fuchs:etal:2014}. Then, constant functions lie in the kernel of $K^X$, $\kernel(K^X)$. This is not recommended if a penalty is used that does not penalize constant functions (as most difference or derivative-based penalties do) to avoid non-identifiability.

\item Penalties with larger null-spaces  increase the  likelihood of a kernel overlap and resulting
non-identifiability problems. For difference or derivative-based penalties, for example, a penalty penalizing deviations from constant functions (first order differences or derivatives) would thus be preferable in this sense to higher-order differences/derivatives. Constant coefficient functions, which then span the penalty null-space, correspond to models with the mean over the functions as covariate - which are often used by practitioners - and thus also lend themselves to intuitive interpretations.

In particular, for penalties where constant coefficient functions span the penalty null\-space, it is straightforward to see that only X-processes that are centered curve-wise will result in a kernel overlap (unless smoothing parameters are estimated to be very small). Unless curve-wise centering is performed as discussed in 2., such processes will typically occur rarely and using first-order difference/derivative penalties should thus guard against many if not most serious identifiability issues in practice.

\end{enumerate}

\paragraph{Countermeasures}

The third point above suggests modifications of the penalty null-space to avoid non-identifiability caused by potential overlap between the penalty null-space $\mathcal{P}_{s\bot}$
and $\kernel(K^X)$. Alternatively, the estimated coefficient surface can be constrained to be orthogonal to functions in the overlap of $\mathcal{P}_{s\bot}$ and $\kernel(K^X)$. We describe three approaches using penalties with empty null-spaces and one constraint-based approach; a systematic comparison of their performance is given in Section \ref{sec:sim:modpen}.
\begin{enumerate}
\item The simplest approach is the use of a simple ridge penalty
$\mP_s=\mI_{K_s}$. However, the resulting estimates will typically not have good
smoothness properties as the ridge penalty is not a roughness penalty in the
conventional sense and its bias towards small absolute values of $\beta(s,t)$
may increase estimation error, cf.~Figures \ref{fig:cca_6_plot},
\ref{fig:sim:identExample} and \ref{fig:mseBeta2}.
\item A second approach uses a modified marginal penalty matrix without null-space
along the lines of the so-called ``shrinkage approach'' described in
\citet{MarraWood2011} and originally developed for the purpose of variable
selection in generalized additive models. \citet{MarraWood2011} replace the marginal penalty
$\mP_s=\mGamma \diag(\rho_1,\dots,\rho_{K_s})\mGamma^T$, with eigenvectors contained  in
$\mGamma$, eigenvalues $\rho_1, \dots, \rho_{K_s}$ and rank $\tilde K_s < K_s$,
by a full rank marginal penalty
\begin{align*}
\tilde\mP_s &= \mGamma \diag(\rho_1,\dots,\rho_{\tilde K_s},
\epsilon \rho_{\tilde K_s}, \dots, \epsilon \rho_{\tilde K_s})\mGamma^T.
\end{align*}
This substitutes the zero eigenvalues $\rho_{\tilde K_s + 1}, \dots, \rho_{K_s}$
with  $\epsilon \rho_{\tilde K_s}$ for all $k = \tilde K_s + 1,\dots,K_s$ using
$0 < \epsilon \ll 1$,
and thereby adds a small amount of penalization to parameter vectors in the null
space of the original penalty. In the following, we will refer to such a modified penalty  as a full-rank penalty. By imposing a small degree of regularization on
functions in $\mathcal P_{s\bot}$ one can in principle preserve the attractive smoothing
properties of the original penalty while still avoiding large artifacts due to
non-identifiability resulting from a kernel overlap. We use $\epsilon=0.1$ as suggested in \citet{MarraWood2011}, but results show that this choice is not always effective in removing artifacts if the estimated smoothing
parameter is very small and overall penalization of the fit is weak.
\item In a similar effort to avoid spurious estimates in scalar-on-function
regression, \citet[][their eq. (16)]{James2005} suggested using the empirical
FPCs of $X(s)$ scaled by their inverse eigenvalues as a penalty. This
penalizes coefficient functions with large variability in
directions in which $X(s)$ varies
very little or not at all (i.e., in $\kernel(K^X)$). We adapted this
approach to our function-on-function setting by replacing the conventional
difference operator based B-spline penalty matrix with
\begin{align*}
\mP_s &= \mB_s^T  \sum^{\min(N,S)}_{m=1} \hat\nu_m^{-1}\diag\left(\mw \cdot \hat\phi_m(\ms)^2\right) \mB_s
\end{align*}


with estimated FPCs and eigenvalues $\hat\phi_m(s)$ and $\hat\nu_m$, $m=1, \dots, \min(N,S)$.
For a given $t_0$, this penalty matrix approximates the marginal penalty term
\mbox{$\sum^{\min(N,S)}_{m=1} \int (\hat\nu^{-1/2}_m\hat\phi_m(s)\beta(s, t_0) )^2 ds$}
suggested by \citet{James2005}.

The empirical $\hat\phi_m(s)$ and $\hat\nu_m$ have to be estimated from a singular value decomposition of $\mX$. It is unclear, however, how to compute the inverse eigenvalues if $\mX$ is of low rank, i.e., if some of the $\hat\nu_m$ are (numerically) zero, which is of course precisely the setting in which this penalty might yield more stable estimates.
In our experiments, we tried replacement of the zero eigenvalues with the smallest non-zero eigenvalue and a variety of other replacement schemes, but results from this approach seem to be fairly sensitive to both the chosen replacement scheme for zero eigenvalues and to the estimated FPCs themselves.

\item Instead of the augmented or alternative penalization schemes for components of the coefficient surface that lie in the overlap $\kernel(K^X) \cap \mathcal{P}_{s\bot}$, we can also directly constrain these components to be zero. As these components are not estimable from the data, such constraints are a plausible default but need to be taken into account for the interpretation of the estimated surface. To implement them, we compute a basis of the overlap and constrain the coefficient surface evaluated on the observed grid to be orthogonal to vectors in the span of that basis. Let $\mX^\top_\bot = {(\mX\tr)}_\bot$ and the $S \times K_s$ matrix $\mB_s$ containing the marginal basis functions over $\mathcal S$ evaluated on $\bm{s}$. A basis spanning the overlap is then defined by the left singular vectors $\mV_{Cs+}$ of the matrix ${\mX^\top_\bot} \left( {(\mX^\top_\bot)\tr} \mX^\top_\bot \right)^{-1}(\mX^\top_\bot)^\top\diag(\mw)\mB_s\bP_{s\bot}$ that have positive singular values. For intuition, consider that if $\mw = \bm{1}$, the expression above projects $\mB_s\bP_{s\bot}$, i.e., a basis for $\mathcal{P}_{s\bot}$, into the span of $\mX^\top_\bot$, i.e, a basis for $\kernel(K^X)$. This yields a basis for the intersection of the spans of $\mB_s\bP_{s\bot}$ and $\mX^\top_\bot$. The expression above is cheap to compute in a numerically stable way using the QR decomposition of $\mX^\top_\bot$. We then estimate $\mTheta$ under the constraints ${\mV_{Cs+}\tr} \diag(\mw)\mB_s \mTheta \stackrel{!}{=} \bm{0}$. To give an example, the constrained coefficient surface estimate for a curve-wise centered functional covariate that does not contain any constant components will be centered around zero and not be offset in any direction if the penalty null-space includes constant functions. Note that, different from the augmented or alternative penalties described above, these constraints are empty if $\kernel(K^X) \cap \mathcal{P}_{s\bot}$ is empty and penalties with constraints then reduce to the usual penalties without constraints. The constraint definition above is quite general and can be used for any basis with a quadratic penalty, it is not limited to the difference penalties we focus on. Specifically, it will also be applicable for derivative based penalties as well as some of the PDE-based penalties introduced by \citet{RamsayHooker2007} whenever identifiability issues arise. By default, the \code{pffr} function uses the diagnostic criterion \eqref{eq:dxp} combined with a check for a rank-deficient design matrix to determine the presence of persistent non-identifiability. If persistent non-identifiability is found, a corresponding warning is issued and the constraints developed here are applied to the fit. Interpretation of the constrained coefficient surface estimates requires careful attention. For example, if constrained surface estimates are centered around zero because $\kernel(K^X) \cap \mathcal{P}_{s\bot}$ contains constant functions, the signs of different regions of $\hat\beta(s,t)$ are not interpretable due to the estimated absolute level of $\hat\beta(s,t)$ being essentially arbitrary.
\end{enumerate}

\section{Simulation study}\label{sec:sim}  
This section presents results on the practical consequences of the theoretical
development in Section~\ref{sec:issue}. Specifically, we
investigate the performance of the tensor product spline-based approach given in
Section \ref{sec:model} on artificial data of varying complexity and noise levels
in terms of estimation accuracy and use the simulation results to validate the
diagnostics for problematic settings we have developed. Subsequently, we present
results for the modified full-rank penalties of \citet{MarraWood2011},  the
FPC-based penalty of \citet{James2005} and the constrained estimates
described in Section~\ref{sec:issue:diag}.
All models were fitted with the \code{pffr()} function
available in the \pkg{refund} package, which estimates the smoothing parameters using
restricted maximum likelihood (REML).

\subsection{Simulation setup}
We simulate data from data generating process~\eqref{eq:model}, with $n=50$ subjects, $\mathcal T= \mathcal S = [0,1]$
and $S=100$ grid-points for $X_i(s) = \sum^M_{m=1} \xi_{im} \phi_m(s)$. The effect surface $\beta(s,t)$ is estimated
using tensor product cubic B-splines. We set $K_t=10$ and use a marginal first order difference penalty for the $t$-direction.
Test runs showed that results are insensitive to the number of grid-points for the response; we used $T=50$
grid-points for $Y(t)$.
Twenty replications were simulated for all sensible combinations of the
following parameters (144000 replicates in total). For the fitting algorithm,
we vary
\begin{itemize}
\item the \textbf{marginal roughness penalties} for the spline coefficients:
either second order difference penalties (``$\Delta^2$'') or first order
difference penalties (``$\Delta^1$'').
For the second order differences, coefficient vectors in the penalty's null-space $\kernel(\mP(\lambda_s,\lambda_t))$
parameterize surfaces that are constant or linear in both directions.
For the first order differences, coefficient vectors in
$\kernel(\mP(\lambda_s,\lambda_t))$ parameterize constant surfaces.
Note that rank  deficiencies can increase if associated smoothing parameters become sufficiently small. In that case, the
penalty null-space effectively becomes the entire span of the associated basis
functions.
\item and the \textbf{number of basis functions} over $\mathcal S$: $K_s \in
\{5, 8, 12\}$.
\end{itemize}
For the data generating process, we vary the following parameters:
\begin{itemize}
\item \textbf{number of eigenfunctions for the $X(s)$-trajectories} with
non-zero eigenvalues:\\  $M \in \{3, 5, 8, 12, 20\}$. This means we have settings with $M\leq K_s$ and $M>K_s$ for
most $M$. Note that the effective numerical rank of a simulated $\mX$
can be (much) lower than $M$ depending on the speed with which the eigenvalues decrease.
\item  \textbf{signal-to-noise ratio}: $\snr_\eps =
\tfrac{\sd(\int_{\mathcal{S}} X_{i}(s)\beta(s,t)ds))}{\sd(\eps_{it})} \in \{2, 10, 1000\}$, where $\sd(x)$
is the empirical standard deviation of $x$. This corresponds to high
and intermediate noise levels for realistic scenarios as well as settings with
almost no noise to check the theoretical properties.
\item \textbf{FPC systems} for $X_i(s) = \sum^M_{m=1}
\phi_m(s)\xi_{im}$ with $\xi_{im} \sim N(0, \nu_m)$, with different patterns of decrease in the eigenvalues $\nu_m$
of the covariance operators: either a linear decrease ($\nu_m=\tfrac{M+1-m}{M}$) or
an exponential decrease ($\nu_m=\exp\left(-\tfrac{m-1}{2}\right)$).
Some of these processes are constructed so that their covariance
operator kernels $\kernel(K^X)$ include functions in the penalty null-space $\mathcal{P}_{s\bot}$.
\begin{itemize}
\item \emph{Poly}: eigenfunctions are orthogonal polynomials of degree 0 to
$M-1$ with  linear or exponentially decreasing eigenvalues
(\emph{Poly,Lin} and \emph{Poly,Exp}, respectively).
For \emph{Poly}, $\kernel(K^X)$ is disjunct from $\mathcal{P}_{s\bot}$, since
the first and second eigenfunctions are constant and linear
polynomials.
\item \emph{Fourier}: eigenfunctions are those of a standard Fourier basis.
Although a complete Fourier basis is a basis for all
square-integrable functions, in practice the kernel $\kernel(K^X)$ of a truncated Fourier basis
contains functions that are very close to the constant since no linear combination of a finite set of
Fourier basis functions yields an exactly constant function, so
$\kernel(K^X)$ is not disjunct from $\mathcal{P}_{s\bot}$.
We used this basis with constant ($\nu_m\equiv 1$ for
\emph{Fourier,Const}) or exponentially decreasing (for \emph{Fourier,Exp})
eigenvalues $\nu_m$.
\item \emph{Wiener}: eigenfunctions and eigenvalues are those of the standard Wiener process on $[0,1]$:
$\phi_m(s) = \sqrt{2}\sin\left(\pi(m - 0.5) s\right)$;
$\nu_m= \left(\tfrac{\pi}{2}(2m + 1)\right)^{-2}$.  The $\kernel(K^X)$
is close to $\mathcal{P}_{s\bot}$ in this case as no linear combination of a finite set of
these basis functions yields an exactly constant or linear function.
\item \emph{BrownBridge}: eigenfunctions and eigenvalues are those of the
standard Brownian bridge on $[0,1]$:
$\phi_m(s) = \sqrt{2}\sin\left(\pi  m  s\right)$;
$\nu_m= \tfrac{1}{\pi m}$. The $\kernel(K^X)$
is close to $\mathcal{P}_{s\bot}$ in this case as no linear combination of a finite set of
these basis functions yields an exactly constant or linear function.
\item \emph{Poly(1+), Poly(2+), Poly(-1)}: eigenfunctions are orthogonal
polynomials of degree 1 [2] to $M$ [$M+1$] for \emph{Poly(1+)}
[\emph{Poly(2+)}],  so that $\kernel(K^X)$ includes the (complete) null-space
of the rank-deficient penalties, i.e., the constant [and linear] functions.
\emph{Poly(-1)} has polynomial eigenfunctions of degree $\{0, 2, 3, 4, \ldots, M+1\}$ so
that $\kernel(K^X)$ overlaps the null-space of the second differences penalty
but not the first differences penalty. All three processes are associated with linearly decreasing $\nu_m$.
\end{itemize}
From top to bottom, these processes become
increasingly more ``antagonistic'' in the sense that 1) the kernels of these
eigenfunction systems move increasingly closer to the kernels of the
penalties we consider and
2) more quickly decreasing eigenvalues result in lower effective rank of
the observed $X(s)$.
\item \textbf{coefficient functions} $\beta(\ms, \mt)=\mB_s \mTheta \mB_t^T$ are drawn
randomly for each setting. The associated coefficients are drawn as
$\ovec(\mTheta) \sim N\left(\bm{0}, \left(0.1 \mI + \mP(\lambda_s,\lambda_t)\right)^{-1}\right)$,
where $\mP(\lambda_s,\lambda_t)$ is a first order difference tensor penalty
matrix.
\begin{itemize}
\item The marginal B-spline bases $B_s$ and $B_t$ have either $4$ or  $8$
basis functions for each direction.
\item $\lambda_s=\lambda_t$ are either $0.1$ or $1$.
\end{itemize}
This generates  coefficient surfaces of varying complexity and roughness. We do not fit models where the basis
used to generate $\beta(s, t)$ is larger than that used to estimate
$\beta(s, t)$ since that could introduce a distracting approximation error not
relevant to the issues at hand.
\end{itemize}
In order to make results comparable across the different settings, we
use the relative integrated mean squared errors
rIMSE$_\beta=\frac{\int(\beta(s,t)-\hat\beta(s,t))^2dsdt}{\int(\beta(s,t))^2dsdt}$
and rIMSE$_Y=\frac{1}{N}\sum^N_{i=1} \frac{\int(\hat Y_i(t)-
\EV(Y_i(t)))^2dt}{\int(Y_i(t) - \overline{Y_i(t)})^2dt}$, where $\overline{Y_i(t)}$
is the mean of $Y_i(t)$ over $t$.

%

\subsection{Results}

\paragraph{Identifiability}

The estimation accuracy for $\hat Y(t)$ (not shown) is excellent across the
board even for the very noisy settings, with 90\%  of relative integrated mean
square errors below 0.01 and a median of 0.00051. Estimation accuracy for
$\hat\beta(s,t)$, however, varies wildly over a range of 18 magnitudes between \ensuremath{1.3\times 10^{-8}}
and \ensuremath{2.2\times 10^{10}}. Further analysis shows that the simulation study design succeeds in creating the identifiability issues described by the results in
Section \ref{sec:issue:eqrep}. To quantify the severity of identifiability issues,
we compute rank correlations between rIMSE$_\beta$ and rIMSE$_Y$ over the 20
replicates of each simulation setting. As expected, we observe low or
even negative correlations mostly for settings in which $\mD_s$ is rank-deficient.
This effect increases both for lower signal-to-noise ratios and for more complex true shapes of $\beta(s,t)$. For intuition, consider that the ``best'' solution for
\eqref{penest} for any given error will be the smoothest surface
(i.e., the one with the smallest penalty term) in the set of surfaces
that can be generated by adding functions from $\kernel(K^X)$ to any initial
$\beta(s,t)$ with the given error.
This may be quite close or quite far from the true $\beta(s,t)$, depending on
the specific setting, with more noisy data and more complex true shapes more likely
to result in fits that are quite far from the truth and still producing good
model fit.

\paragraph{Estimation performance for $\beta(s,t)$}
\begin{figure}[!ht] {\centering}
\begin{knitrout}\tiny
\definecolor{shadecolor}{rgb}{0.969, 0.969, 0.969}\color{fgcolor}

{\centering \includegraphics[width=.95\textwidth]{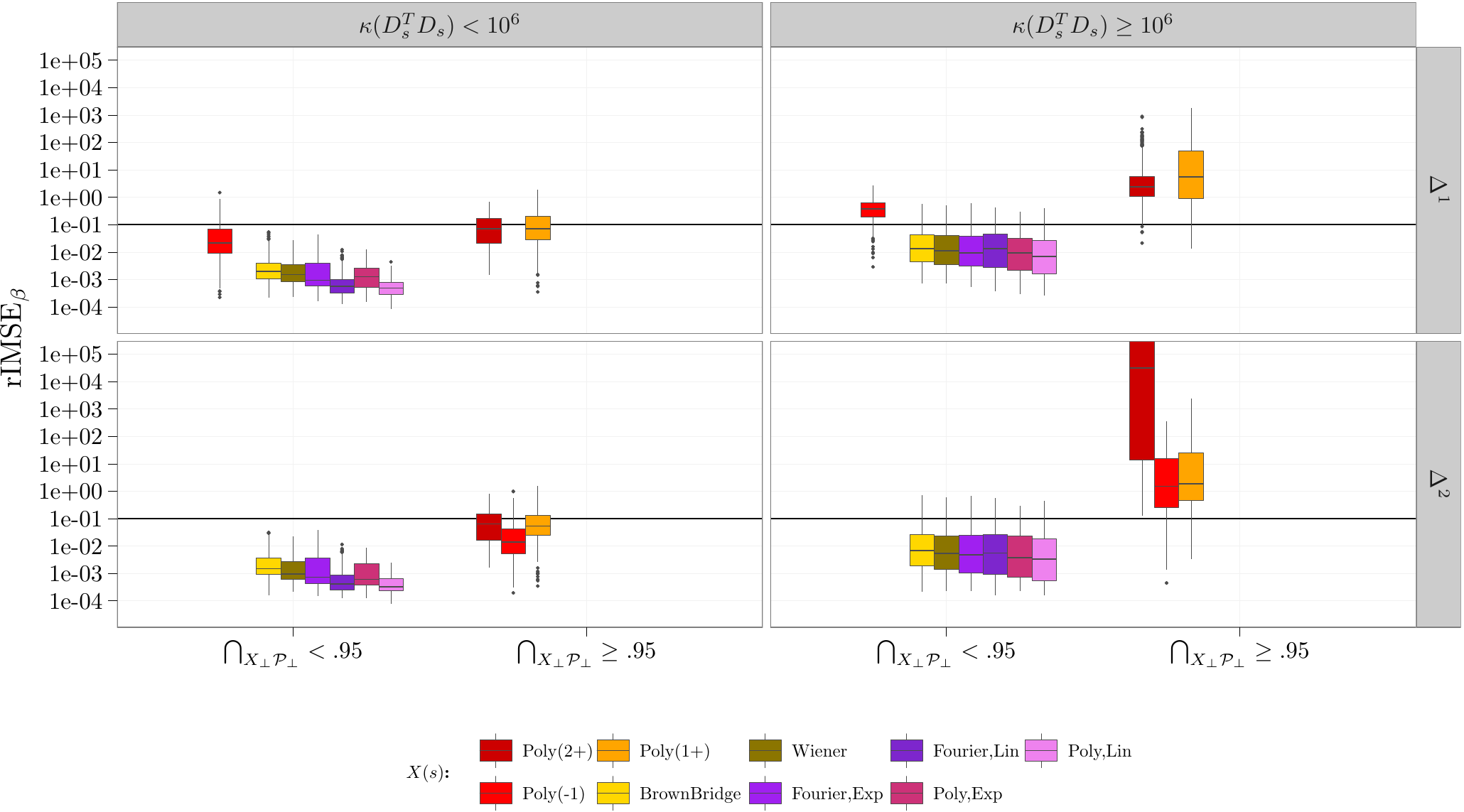} 

}

\end{knitrout}
\caption{\label{fig:mseBeta} Boxplots for relative integrated mean square
error  rIMSE$_\beta$ for all 18000 results
for SNR$_\epsilon=10$ with conventional difference penalties.
Columns show results for full rank $\mD_s$ versus numerically rank deficient
$\mD_s$. Rows show results for the first and second order difference penalties.
Boxplots grouped by overlap between $\kernel(K^X)$ and $\mathcal{P}_{s\bot}$ as
computed by $\bigcap_{X_\bot \mathcal{P}_{s\bot}}$, color-coded for the
different processes the $X(s)$ are sampled from.
Vertical axis on $\log_{10}$-scale, extreme errors $> 10^5$ for \emph{Poly(2+)}
are cut off.}
\end{figure}
Figure~\ref{fig:mseBeta} shows the estimation errors for coefficient surfaces
generated with SNR$_\epsilon=10$.
The right column shows results for numerically rank deficient $\mD_s$
(i.e., condition number $\kappa(\bm D_s\tr\bm D_s) \geq 10^6$),
the left column for designs with $\kappa(\bm D_s\tr\bm D_s)<10^6$.
The top row shows results for first differences penalty, bottom row for
second differences penalty.
Boxplots are grouped by the amount of overlap between $\kernel(K^X)$ and $\mathcal{P}_{s\bot}$ as  computed by $\bigcap_{X_\bot \mathcal{P}_{s\bot}}$ (see \eqref{eq:dxp}), color-coded for the different processes the $X(s)$-trajectories are sampled from.
Results for SNR$_\epsilon=2$ and $10^3$ were qualitatively very similar -- errors obviously become larger for noisier data
but the pattern shown in Figure~\ref{fig:mseBeta} remains the same.
Note that relative estimation errors below $\approx 0.01$ correspond to
estimates that are visually indistinguishable from the true surfaces,
and that errors below $\approx 0.1$ (thick black horizontal line)
usually preserve most essential features of the true $\beta(s,t)$ well. Results
with rIMSE$_\beta>1$ bear little resemblance to the ``true'' function.
Also recall that all of these fits, with rIMSE$_\beta$ values varying by more
than 14 magnitudes,
resulted in a comparatively small range of rIMSE$_Y$ between
$10^{\ensuremath{-5}}$ and
$10^{\ensuremath{-3}}$.
Closer inspection of results shows that the extremely
large errors for \emph{Poly(1+), Poly(-1)} and \emph{Poly(2+)} are caused by the expected
behavior: estimates are shifted by functions from the overlap of $\kernel(K^X)$
and $\mathcal{P}_{s\bot}$. The top row of Figure~\ref{fig:sim:identExample} shows
an example of this behavior: the estimate for first order difference penalty is
shifted by a large constant, while the estimate for the second order
difference penalty is shifted by both a constant and a huge linear trend in
$s$-direction. The fitted values of all models shown in Figure
\ref{fig:sim:identExample} are practically identical.

These results mostly corroborate the results derived in Section
\ref{sec:issue} -- we see that:
\begin{itemize}
\item Serious errors rIMSE$_\beta>0.1$ are rare for both full-rank and
rank-deficient $\mD_s$ if the generating process for $X(s)$ is not antagonistic in the sense that $\kernel(K^X) \cap
\mathcal{P}_{s\bot} = \{0\}$, i.e. for the \emph{Poly}, \emph{Fourier},
\emph{Wiener} and \emph{BrownBridge} processes.
\item As long as $\kernel(K^X) \cap \mathcal{P}_{s\bot} = \{0\}$ (approximated numerically by the criterion that $\bigcap_{X_\bot \mathcal{P}_{s\bot}} < 0.95$), regularization
of the estimated coefficient surface allows us to achieve good estimates even
if the unpenalized regression model \emph{per se} would not be identifiable due to rank deficiency of $\mD_s$.
\item The larger $\mathcal{P}_{s\bot}$ (top to bottom), and the closer
$\kernel(K^X)$ is to $\mathcal{P}_{s\bot}$ (left to right in each group of boxplots),
the larger the likelihood of estimates far from the truth and the larger the average rIMSE$_\beta$.
\end{itemize}

\paragraph{Diagnostics}\label{sec:sim:diagnostics}
As in Figure \ref{fig:mseBeta}, we use $\bigcap_{X_\bot \mathcal{P}_{s\bot}}$ (see~\eqref{eq:dxp})
as a measure of the overlap between $\kernel(K^X)$ and $\mathcal{P}_{s\bot}$. We consider a replicate to be  ``flagged'' as problematic if both $\bigcap_{X_\bot \mathcal{P}_{s\bot}}\geq 0.95$ and $\kappa(\mD_s^T, \mD_s)>10^6$.
\begin{figure}
\begin{knitrout}\tiny
\definecolor{shadecolor}{rgb}{0.969, 0.969, 0.969}\color{fgcolor}

{\centering \includegraphics[width=.6\textwidth]{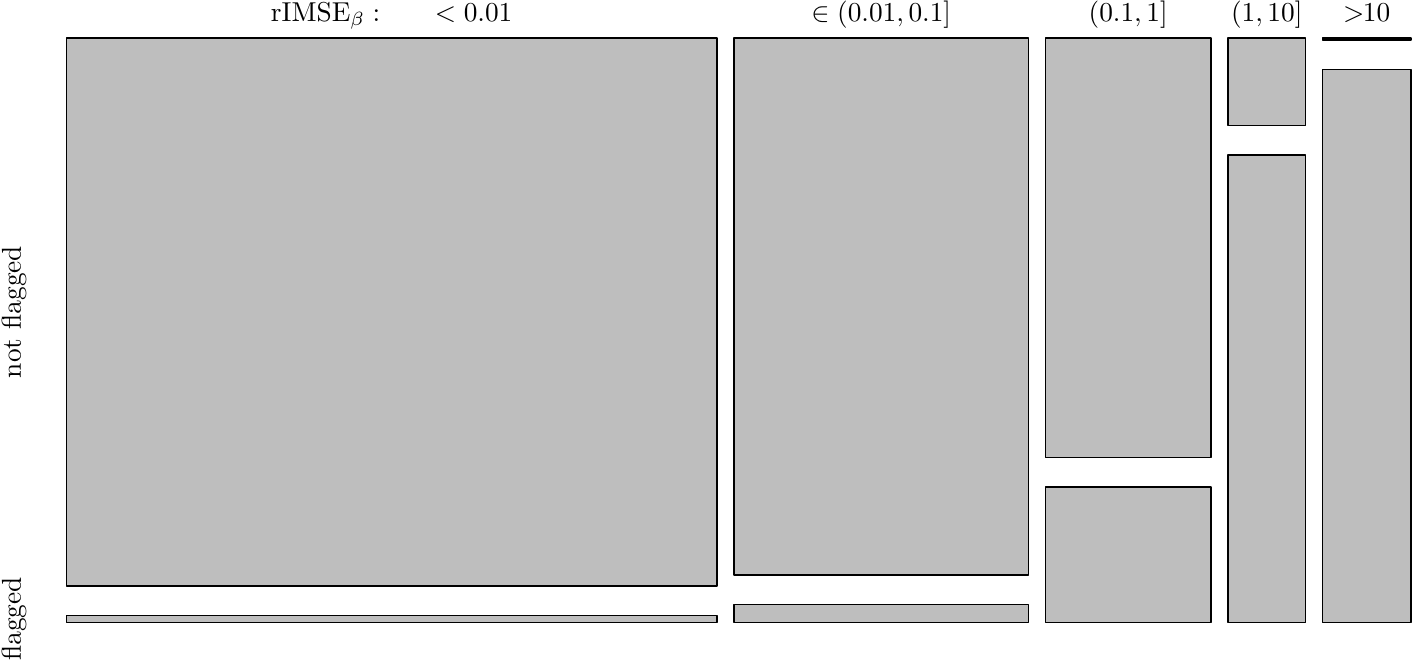} 

}

\end{knitrout}
\caption{\label{fig:flagTable} Mosaic plot for contingency table of "flagged" replicates and and categorized rIMSE$_\beta$.}
\end{figure}

Figure~\ref{fig:flagTable} shows a mosaic plot of the contingency table of ``flagged`'' replicates and categorized rIMSE$_\beta$.
While the sensitivity for identifying replicates with rIMSE$_\beta>1$ is
0.92, the specificity is
only 0.28.
The sensitivity for identifying replicates with rIMSE$_\beta>0.1$ is
0.58, the specificity is
0.09.
These fairly low specificities indicate that the penalized approach to function-on-function regression discussed here can outperform theoretical expectations and frequently finds good solutions even in very difficult settings.
Total accuracy for identifying settings with rIMSE$_\beta>0.1$ is
0.88 and
0.94 for
rIMSE$_\beta>1$.
The positive predictive value (precision) of the criterion for rIMSE$_\beta>1$ is
0.72,
while it is 0.91 for
rIMSE$_\beta>0.1$.
The negative predictive value of the criterion for rIMSE$_\beta>1$ is
0.99,
while it is 0.87 for
rIMSE$_\beta>0.1$.
Also note that certainly not all errors in the $(0.1, 1]$-range are due to
identifiability issues, so not all of the (rare) non-detections are failures of
the criterion.
\begin{figure}[htbp]
\begin{center}

\begin{knitrout}\tiny
\definecolor{shadecolor}{rgb}{0.969, 0.969, 0.969}\color{fgcolor}

{\centering \includegraphics[width=.95\textwidth]{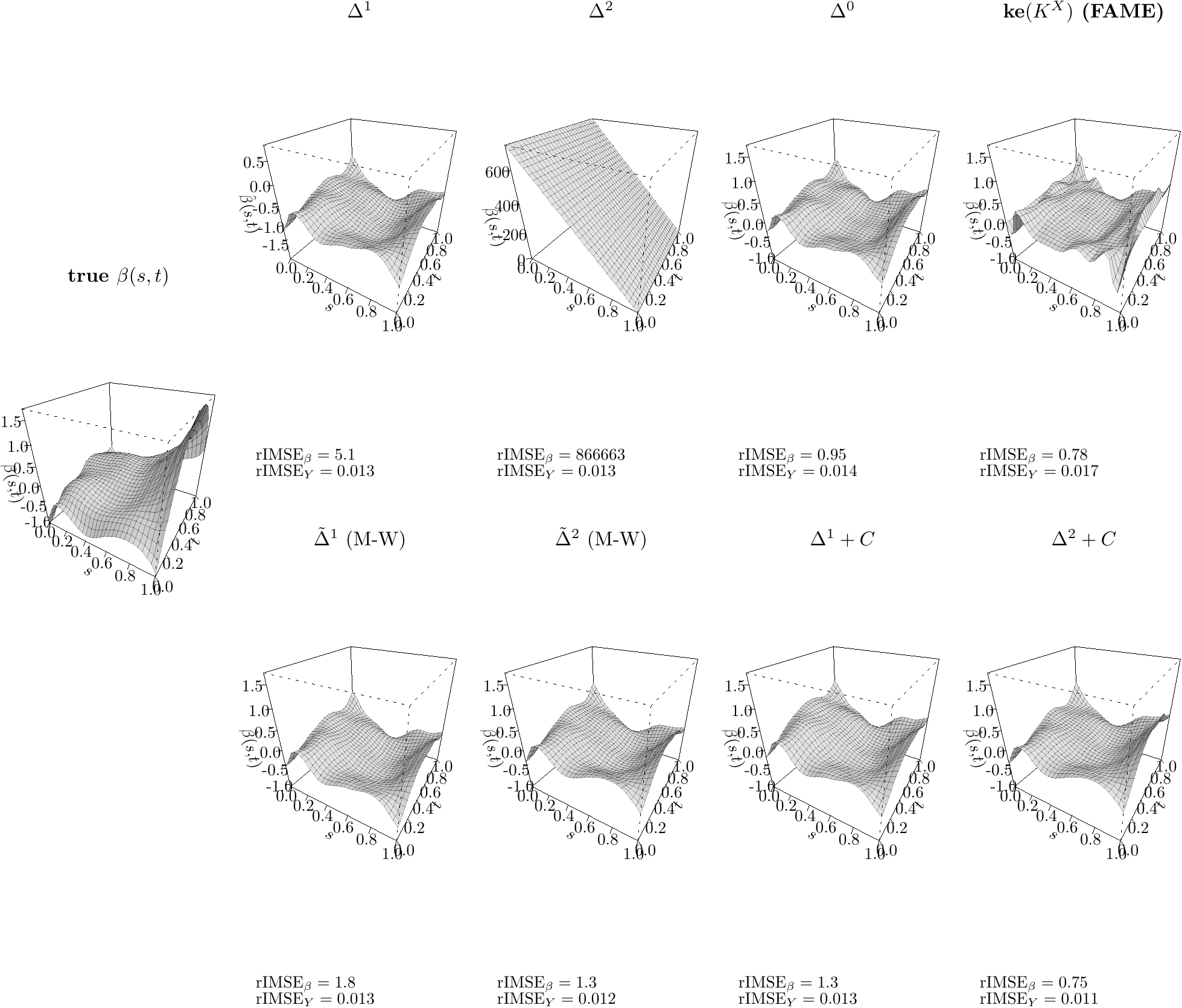} 

}

\end{knitrout}
\caption{Example estimates for different penalties in a very difficult setting with
$X(s)$ from \emph{Poly(2+)}, $M=5$, $K_s = 12$, $\snr_\eps = 2$.
Center left panel: true coefficient surface.
Top row, left to right: first order difference penalty, second order difference penalty, ridge penalty, FAME penalty.
Bottom row, left to right: full-rank first and second order difference penalties,
first and second order difference penalties with constraints.
\emph{Note the different $z$-axis scales in the two left panels of the top row.}
Subtitles give rIMSE$_\beta$ and rIMSE$_Y$ for each fit. }
\label{fig:sim:identExample}
\end{center}
\end{figure}

\paragraph{Performance of modified penalties}\label{sec:sim:modpen}
We broadened the scope of our simulation study by additionally comparing the performance of the non-standard penalties and constraints introduced in Section~\ref{sec:issue:diag}:
\begin{itemize}
\item a full-rank ridge penalty (``$\Delta^0$''),
\item the modified full-rank roughness penalties as suggested by \citet{MarraWood2011}; in our case we used both full-rank first order differences penalties (``$\tilde\Delta^1$'') and full-rank second order differences penalties (``$\tilde\Delta^2$''),
\item the FPC-based penalty of \citet{James2005} (``ke$(K^X)$ (FAME)''). We replace $\hat\nu_m$ by $\max(\hat\nu_m, 10^{-10}\hat\nu_1)$ in order to remove any (numerically) zero or negative eigenvalues,
\item and conventional first and second order differences penalties with additional constraints (``$\Delta^1+C$``, ``$\Delta^2+C$``) that enforce orthogonality of the estimated coefficient surface to functions in $\kernel(K^X) \cap \mathcal{P}_{s\bot}$ if the data are flagged as problematic by our diagnostic criterion ($\kappa(D_s^T D_s)\geq 10^6$ and $\bigcap_{X_\bot \mathcal{P}_{s\bot}}>.95$). Note that $\Delta^1+C$ and $\Delta^2+C$ reduce to $\Delta^1$ and $\Delta^2$, respectively, if $\kernel(K^X) \cap \mathcal{P}_{s\bot} = \emptyset$ and need not be compared to the other techniques in that case.
\end{itemize}

Figure~\ref{fig:mseBeta2} shows the rIMSE$_\beta$ for  $\snr_\eps = 10$ for the different $X(s)$-processes and penalties. Note that, in contrast to Figure~\ref{fig:mseBeta}, colors now represent the different penalties, not the different $X(s)$-processes. Boxplots for $\Delta^1$ and $\Delta^2$ contain the same results as those shown in Figure~\ref{fig:mseBeta}.

The full-rank difference penalties $\tilde\Delta^1$ and $\tilde\Delta^2$ (cyan and turquoise boxes) seem to drastically reduce the size and likelihood of severe estimation errors in the difficult settings, especially compared to $\Delta^2$ (dark blue). We occasionally incur slightly worse estimates for the easier settings, but these differences are small and hardly relevant here. However, we have found that neither $\tilde\Delta^1$ nor $\tilde\Delta^2$ with $\epsilon = .1$ are effective in removing non-identifiability artifacts if the estimated smoothing parameter is very small, i.e., if the effect of the penalty on the fit is weak. This explains the large outliers observed for some replicates for the antagonistic $X$-processes (top row) for $\tilde\Delta^1$ and $\tilde\Delta^2$. Figure \ref{fig:cca-c6-plot} (bottom row, leftmost two panels) shows this for the \code{dti} data set, where the estimated smoothing parameters are small. Note that no such outliers occur for the ridge ($\Delta^0$) and FAME penalties (green boxes). However, both penalties are not competitive for numerically rank deficient $\mD_s$ for the non-pathological $X(s)$-processes in the bottom two rows and perform worse than the (modified) difference penalties in all other settings, and so cannot be recommended for general use as well. Furthermore, in the simulated settings we use, the true
$\beta(s,t)$ are centered around 0, which reduces the negative effect of these penalties' bias towards small absolute values of $\beta(s,t)$.

Figure~\ref{fig:mseBeta3} shows results only for simulated data sets flagged as problematic by our diagnostic criterion for $\Delta^1$ or $\Delta^2$. In such settings, it makes sense to use difference penalties combined with "orthogonal-to-null-space-overlap" constraints on the fitted surface, these results are denoted by $\Delta^1+C$ or $\Delta^2+C$ (violet and purple), respectively.
\begin{figure}[htbp]
\begin{center}
\begin{knitrout}\tiny
\definecolor{shadecolor}{rgb}{0.969, 0.969, 0.969}\color{fgcolor}

{\centering \includegraphics[width=.95\textwidth]{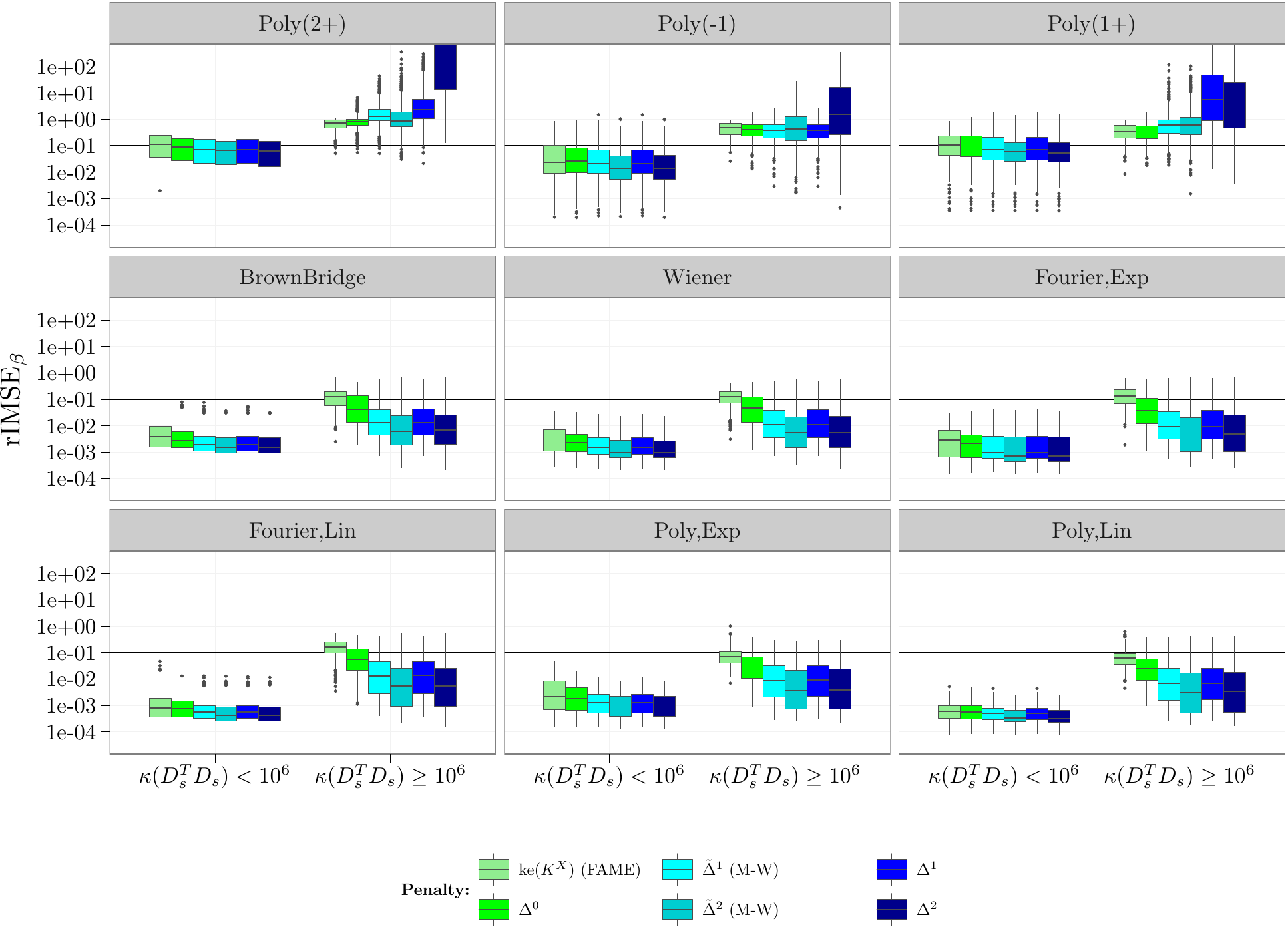} 

}

\end{knitrout}
\caption{rIMSE$_\beta$ for  $\snr_\eps = 10$ for the 8 $X(s)$-processes (panels) and the 8 different penalties (color).
Separate boxplots in each panel for settings with numerically rank deficient $\mD_s$ ($\kappa(\mD_s^T\mD_s)\geq 10^6$) versus settings with full rank $\mD_s$.}
\label{fig:mseBeta2}
\end{center}
\end{figure}

\begin{figure}[htbp]
\begin{center}
\begin{knitrout}\tiny
\definecolor{shadecolor}{rgb}{0.969, 0.969, 0.969}\color{fgcolor}

{\centering \includegraphics[width=.95\textwidth]{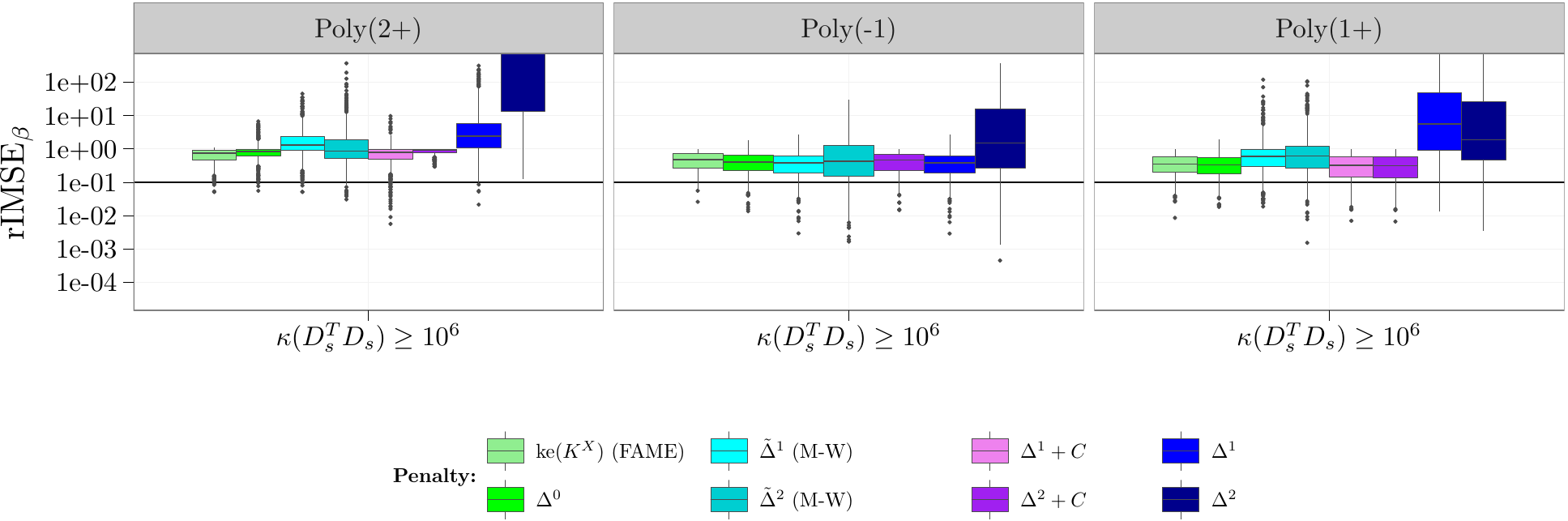} 

}

\end{knitrout}

\caption{rIMSE$_\beta$ for  $\snr_\eps = 10$ only for ``pathological'' data-sets. $X(s)$-processes in panels, penalties coded by color.}
\label{fig:mseBeta3}
\end{center}
\end{figure}

Figure~\ref{fig:sim:identExample} shows illustrative exemplary fits for the modified penalties and penalties with constraints in a data setting with $\kappa(\mD_s^T \mD_s)\geq 10^6$.
in the two rightmost panel of the top row and the panels in the bottom row.
While rIMSE$_Y$ is very similar for all 8 penalties except FAME (top right), the shapes of the corresponding coefficient surface estimates differ from each other in the expected fashion:
The ridge penalty's bias towards small $|\beta(s,t)|$ and tendency to under-smooth the estimated surface is visible, as is the typical wiggliness of FAME fits which are also not roughness penalized in the conventional sense. 
Compared to the result for $\Delta^1$, the fit for $\tilde\Delta^1$ is much closer to the level of the true $\beta(s,t)$ and adds a much smaller (spurious) constant to the fit.
Similar remarks apply for the comparison between $\Delta^2$ and $\tilde\Delta^2$: The huge spurious linear trend in $s$ is almost completely removed. The overlap criterion $\bigcap_{X_\bot \mathcal{P}_{s\bot}}$ is $\approx 1$ for $\Delta^1$, $\approx 2$ for $\Delta^2$ and 0 by construction for all others. Both constrained fits ($\Delta^1+C$, $\Delta^2+C$, two  bottom right panels) completely suppress the spurious constant (and trend) present in the unconstrained fits (two top right panels), and $\Delta^2+C$ yields a slightly smoother fit than $\Delta^1+C$ as expected.

Note that the true $\beta(s,t)$ are centered around 0 in our simulation. As all modified penalties and penalties with additional constraints result in estimated surfaces around 0, results will be shifted up- or downwards compared to the true surfaces if this is not the case in a real application. Such a shift reflects the fact that the average level of the surface cannot be inferred from the data in settings with corresponding persistent non-identifiability. Thus, in cases where our diagnostics indicate persistent non-identifiability, estimated coefficient surface values can only be interpreted relative to one another, but no interpretation of their sign is possible due to the estimated absolute level being essentially arbitrary. A corresponding implication regarding linear trends applies to the case of second order difference penalties.

\paragraph{Summary of simulation results}
We can draw the following conclusions based on the entirety of simulation
results:
\begin{itemize}
\item the potential for extreme estimation errors is large for $X(s)$-processes
with low effective rank whose $\kernel(K^X)$ is not disjunct from
$\mathcal{P}_{s\bot}$.
\item there is no strong positive correlation between accuracy of fitted values
(rIMSE$_Y$) and accuracy of the estimated coefficient surface (rIMSE$_\beta$)
for rank deficient $\mD_s$. Even extremely wrong
estimates of $\beta(s,t)$ can yield good model fits.
\item calculating $\bigcap_{X_\bot \mathcal{P}_{s\bot}}$ and $\kappa(\mD_s^T\mD_s)$ yields a suitable criterion that can diagnose persistent non-identifiability reliably, albeit with a substantial rate of "false alarms" in which conventional difference penalties work well despite an antagonistic data situation.
\item the full-rank roughness penalties often stabilize estimates in persistent non-identifiability settings without diminishing accuracy in other settings by a relevant amount, but not reliably so.
\item for problematic data settings flagged by the criterion developed here, the combination of difference penalties with suitable constraints on the coefficient surface ("$\Delta + C$"), ridge penalties ("$\Delta^0$") and penalties based on the spectrum of the functional covariate ("FAME") all perform similarly and outperform both full-rank and conventional difference penalties.
\end{itemize}

\section{Case Study}\label{sec:example}
\begin{figure}[!ht] {\centering}
\includegraphics[width=\textwidth]{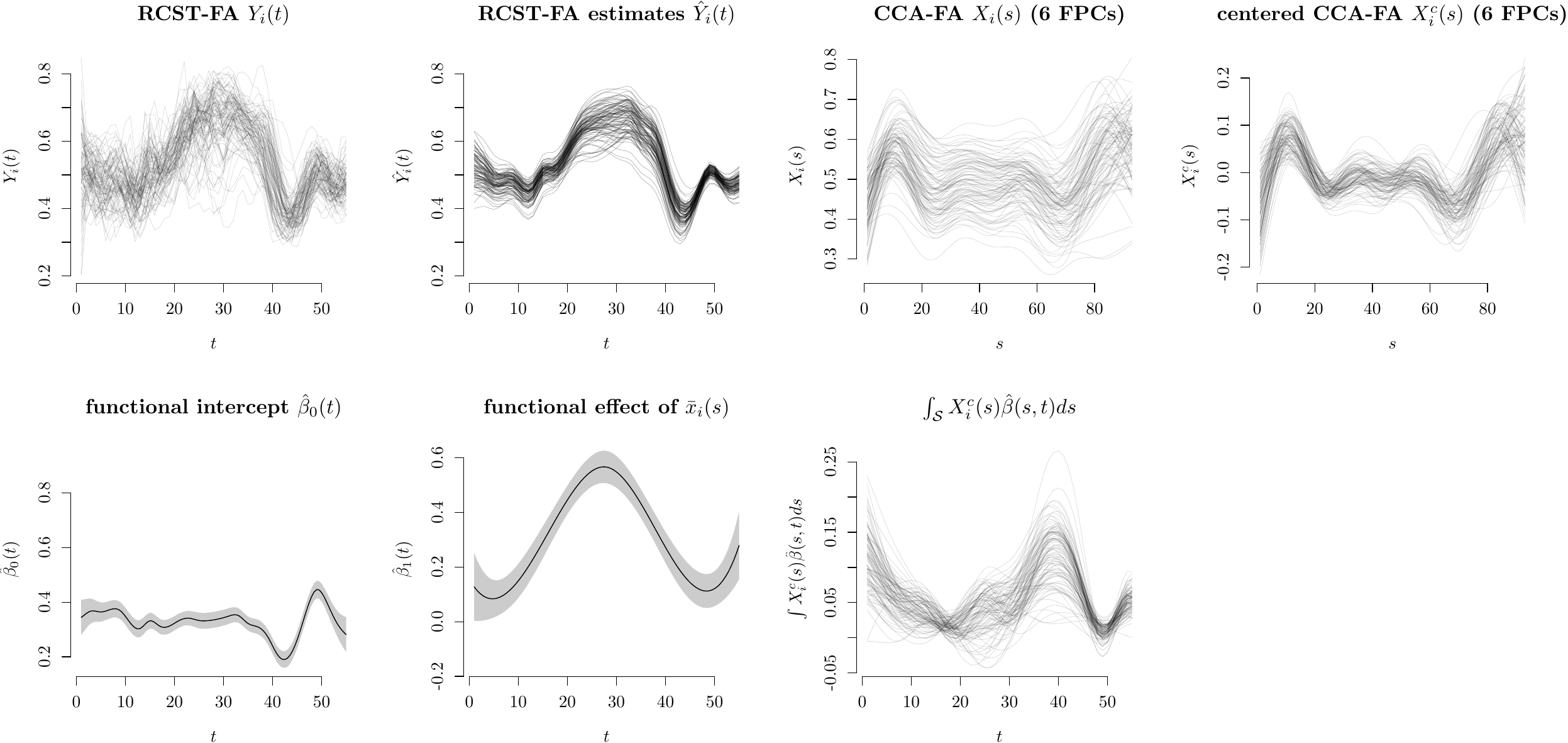}
\caption{\label{fig:cca-c6-data-fits} Top row, left to right: Observed RCST-FA $Y_i(t)$;
estimated RCST-FA $\hat Y_i(t)$ for the model described in this Section; presmoothed CCA-FA based on 6 largest FPCs as used in the example shown in Figure \ref{fig:cca_6_plot};
presmoothed, curve-wise centered CCA-FA based on 6 largest FPCs as used in the example below. Bottom row, left to right: Estimated functional intercept $\hat\beta_0(t)$ and 
estimated functional effect of mean CCA-FA $\hat\beta_1(t)$, with approximate point-wise 95\% confidence intervals; estimated contributions of the functional effect $\int_{\mathcal S} X^c_i(s) \hat\beta(s,t) ds$ to the additive predictor.}
\end{figure}

\begin{figure}[!ht] {\centering}
\includegraphics[width=\textwidth]{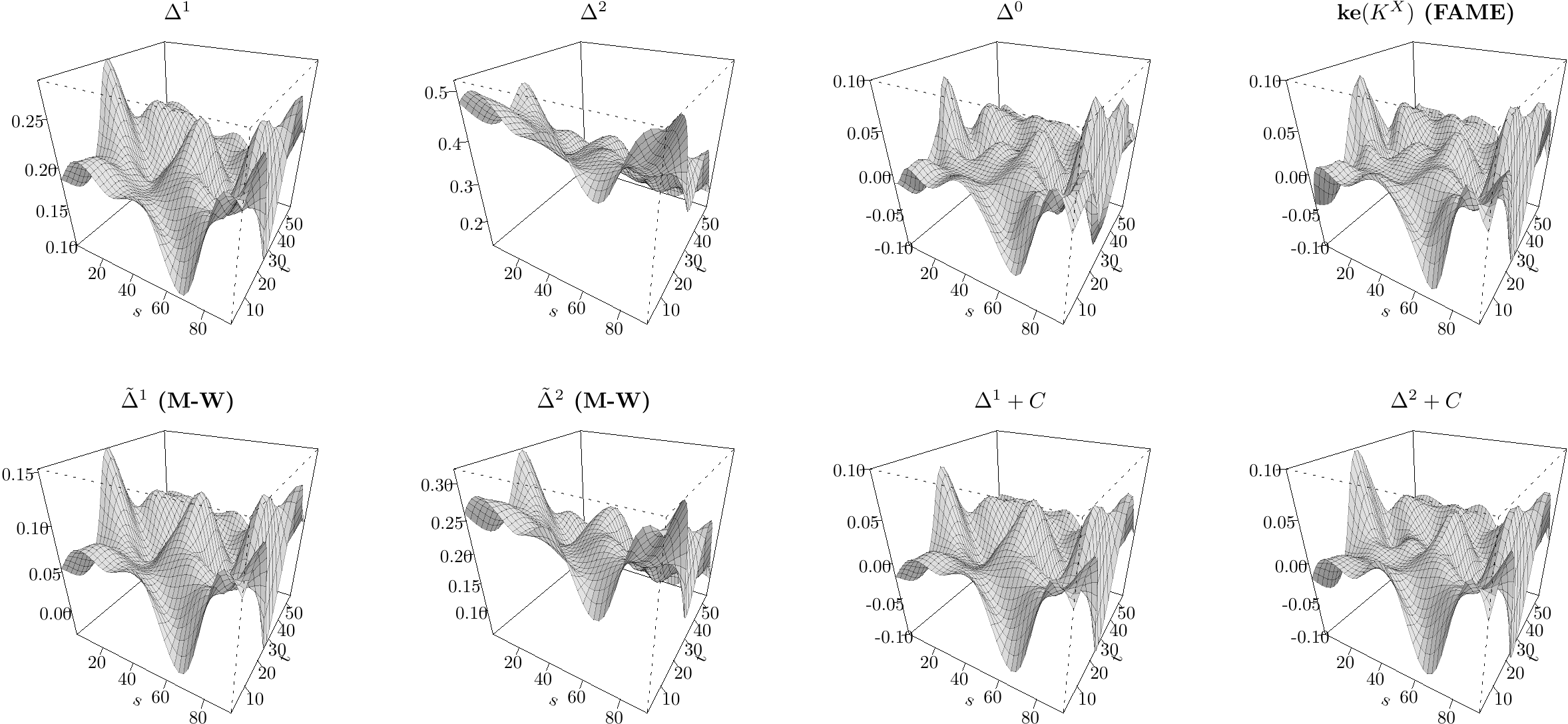}
\caption{\label{fig:cca-c6-plot} Estimated coefficient surfaces $\hat\beta(s,t)$ for regressing RCST-FA on mean CCA-FA and centered CCA-FA truncated to its first 6 empirical FPCs. All 8 fits lead to practically identical fitted values. Top row from left: First order ($\Delta^1$), second order difference penalty ($\Delta^2$), ridge penalty ($\Delta^0$), FPC-kernel-based penalty (FAME). Bottom row from left: Full rank first order ($\tilde\Delta^1$ (M-W)), second order difference penalty ($\tilde\Delta^2$ (M-W)), first order with constraint on kernel overlap ($\Delta^1 + C$), second order with constraint on kernel overlap ($\Delta^2 + C$). All fits performed with \code{pffr()}. \emph{Note that z-axis limits for the 4 panels on the left differ from each other and from those of the 4 panels on the right.}}
\end{figure}

To illustrate the practical relevance of the theoretical development given in Section \ref{sec:issue} and the performance of the countermeasures developed therein, we describe a deliberately constructed but realistic setting in which non-identifiability persists despite penalization. We use a subset of 100 patients from the \code{dti}-data described in the introduction and fit a realistic model in which investigators try to separate the effect of mean CCA-FA levels on RCST-FA from that of the shape of CCA-FA along the tract.
Such a model can be defined as
$$
Y_i(t) = \beta_0(t) + \bar x_i \beta_1(t) + \int_\mathcal{S} X^c_i(s) \beta(s,t) ds + \epsilon_{it}; \qquad \epsilon_{it} \stackrel{\iid}{\sim} N(0, \sigma^2),
$$
where $Y_i(t)$ is RCST-FA for patient $i$, $\beta_0(t)$ is a global functional intercept, $\bar x_i = \overline{X_i(s)}$ are the mean
CCA-FA levels with associated effect $\beta_1(t)$, and $X^c_i(s) \approx X_i(s) - \bar x_i$ are the denoised, curve-wise centered CCA-FA curves. In this example, we use a reconstruction $X^c_i(s)$ based on the six largest empirical FPCs (ca. 90 \% of total variance explained) to denoise CCA-FA and interpolate missing values. Note that curve-specific mean centering causes all constant functions to be in the kernel of the covariance of $X^c_i(t)$, i.e. $\bigcap_{X_\bot \mathcal{P}_{s\bot}} \geq .95$ for first and second order difference penalties. All fits shown below use $K_s=K_t=12$ marginal basis functions, so by construction we expect $\kappa(\mD_s^T\mD_s)$ to be large since the rank of the functional covariate is at most $M=6$. More than 6 marginal basis functions are required, however, so that the basis for $\beta(s,t)$ is flexible enough for approximating the rather complex shape we observe here. Note that curve-wise centering is also a necessary pre-processing step for many spectroscopic data analyses where absorption spectra are used as functional covariates and different mean intensity levels of such spectra are often pure laboratory artifacts \citep{Fuchs:etal:2014}, and more generally in settings such as this one where practitioners try to separate effects of the mean level of a functional covariate from those of its shape.

From left to right, the top row of Figure~\ref{fig:cca-c6-data-fits} shows the responses and fitted values for this model, the uncentered functional covariates used for the models shown in Figure~\ref{fig:cca_6_plot}, and the centered functional covariates used for the models in this section. For uncentered covariates, this setting provides an example of \emph{simple} non-identifiability, while centered covariates induce \emph{persistent} non-identifiability for difference penalties.
Figure~\ref{fig:cca-c6-plot} shows the estimated coefficient surfaces $\hat\beta(s,t)$ for this model for various penalties. Despite the divergent shapes of the different coefficient surface estimates, all 8 fits lead to practically identical fitted values on the training data (100 patients) and practically identical predictions for the test set (155 patients), with integrated MSE $\approx 0.0031$ and predictive integrated MSE $\approx 0.0046$ for each model. Both $\hat\beta_0(t)$ and $\hat\beta_1(t)$ are also practically identical for all 8 models, they are shown in the bottom row of Figure~\ref{fig:cca-c6-data-fits} along with the estimated contributions $\int_{\mathcal S} X_i^c(s) \hat\beta(s,t)ds$ of the functional covariate to the additive predictor.

The model specification is appropriately flagged as causing \emph{persistent} non-identifiability (i.e., $\bigcap_{X_\bot \mathcal{P}_{s\bot}} = .98\,[1.15]$ and $\kappa(\mD_s^T\mD_s) \rightarrow \infty$) for first [second] order difference penalties $\Delta^1$ [$\Delta^2$]. It is instructive to compare the resulting estimates for $\beta(s,t)$ for different penalty specifications: the first difference penalty fit (top, left) includes a constant offset from 0, while the second difference penalty fit (top, second from left) includes different offsets from 0 for each $t$, in a linear decrease that is unpenalized by the marginal second order difference penalty in $t$ direction, as well as a linear trend in $s$. The presence of this spurious linear trend is surprising since we did not explicitly remove linear components from the functional covariate and it is not present in the fits for the uncentered data (c.f. Figure~\ref{fig:cca_6_plot}). Further investigation revealed that the centered functional covariate's null-space contains a fairly strong linear component ($\bigcap_{LV}((\bm{X}^c)\tr_\bot, \mw \cdot \bm s) = 0.77$, where $(\bm{X}^c)\tr_\bot$ denotes the orthogonal complement of the observed $X^c(s)$), which is overlap enough to produce spurious linear shifts in this example. Note that the linear component is much stronger in the non-centered covariates ($\bigcap_{LV}((\bm{X})\tr_\bot, \mw \cdot \bm s) = 0.11$) and its lack is caused by the curve-wise centering in this example. It is straightforward to show that this can occur if $\int X_i(s) s ds \approx \bar x_i \int s ds \;\forall i = 1,\dots,n$, as is the case here.

Both offset and trends are reduced, but not entirely removed, by instead using the full-rank difference penalties with $\epsilon = 0.1$ (bottom row, two leftmost panels). Due to the wiggliness of the coefficient surface, the estimated smoothing parameters are quite small, so there is only little penalization going on. Consequently, the additional small penalty on $\mathcal{P}_{s\bot}$ is not strong enough to eliminate non-identifiability artifacts from the fit in this case. Estimated coefficient surfaces with full-rank difference penalties with $\epsilon = 1$ (not shown), however, are very similar to the those in the four right-most panels of Figure~\ref{fig:cca-c6-plot}. That results for the full-rank difference penalties depend so strongly on this tuning parameter is a considerable disadvantage. 

The four right-most panels show results based on the ridge penalty ($\Delta^0$, second from right, top row), the penalty based on the FPCs of $\kernel(K^X)$ (FAME, top right), and the difference penalties with additional ``orthogonal-to-kernel-overlap''  constraints ($\Delta + C$, bottom row). It is reassuring to see that all four of these penalties lead to very similar results in this setting despite their different motivations and mathematical properties. Admittedly, understanding the estimated coefficient surfaces in terms of the supposed data generating process remains difficult. First, because of their complex shape and secondly (and more pertinently to the main points of this paper), because the restriction to surfaces centered around 0 that is enforced for $\Delta + C$ and implied by the $\Delta^0$ and $\kernel(K^X)$ penalties precludes facile interpretation of, e.g., the sign of features of $\beta(s,t)$. Since the ``true'' offset of the coefficient surface is not estimable from the data, negativity or positivity of certain peaks or troughs of $\hat\beta(s,t)$ is not directly interpretable and interpretation can thus only rely on relative heights across the surface. Higher dimensional overlaps between penalty and covariate null-spaces than the one encountered here will compound these expositional difficulties.

This example demonstrates how reasonable, but unfortunate combinations of data pre-processing and model specifications can lead to simply as well as persistently non-identifiable models despite penalization. The diagnostics and countermeasures developed in Section \ref{sec:issue}, however, seem to be suitable for detecting and remedying such problems in a real data setting.

\section{Conclusion and Discussion}
Coefficient surface estimates in spline-based function-on-function-regression \eqref{eq:model} can suffer from persistent identifiability problems if the span of the marginal basis for the coefficient surface over a functional covariate's domain overlaps the kernel of its covariance operator. 
A rank deficient  design matrix can occur in  particular if the functional covariate's covariance operator is of effective rank smaller than the number of marginal basis functions - 
either because the number of eigenfunctions with non-zero eigenvalues is truly below the number of marginal basis functions, or if eigenvalues of the covariance operator decrease too rapidly compared to the noise level of the data.

In practice, spline based approaches are typically fitted with a regularization penalty corresponding to a smoothness assumption on the coefficient surface. 
We have shown that identifiability problems persist if, and only if, in addition to a rank deficiency of the design matrix, the kernel of the functional predictor's covariance $\kernel(K^X)$ overlaps the function space $\mathcal P_{s\bot}$ spanned by parameter vectors in the null-space of the spline's roughness penalty. 
In the case of no overlap, there is a unique smoothest representative on any hyperplane of parameter vectors leading to the same additive predictor and the penalized estimation problem finds a unique smoothest solution.
Similar results hold for the simpler case of penalized scalar-on-function regression models \citep{Happ2013}. They are also expected to hold for more general loss functions than the quadratic loss analysed here since generalized additive models are typically estimated by the penalized iteratively re-weighted least squares (P-IWLS) method, where the system of equations solved in each step is identical to that of \eqref{penest} except for the introduction of a vector of IWLS weights and the substitution of $\bm y$ by IWLS working responses that are element-wise linear transformations of $\bm y$.

A lack of identifiability also implies a lack of correlation between accuracy of the coefficient estimates and goodness of fit for the responses. 
As this extends to prediction errors for out-of-sample data from the same process, it is usually not possible to detect  identifiability issues for a given data set based on subsampling or cross-validation schemes. 
Instead, based on theoretical considerations and simulation results, we have identified two easily computable diagnostic criteria in order to detect non-identifiable model specifications before estimation. 
The criteria combine  the condition number of a partial design matrix with a measure of the amount of overlap between the kernel of the functional predictor's covariance and the null-space of the penalty. Non-identifiability may in particular be an issue if both criteria are indicative of a problematic setting, or if the partial design matrix is numerically rank deficient and the penalty smoothing parameter is estimated to be close to zero. 
 If a non-identifiable model specification is discovered, we recommend that practitioners choose  modified full-rank roughness penalties to safeguard against
spurious estimates or estimate coefficient surfaces under constraints that force these spurious components to zero. The \texttt{pffr}-function in the refund package uses a first order differences penalty by default, incorporates the diagnostic checks developed and evaluated in this work, and issues corresponding user warnings and enforces suitable constraints if a persistently non-identifiable model specification is detected.

Another practical consequence of our results is that pre-processing methods for functional covariates should be avoided if they reduce their effective rank (such as pre-smoothing with low-dimensional bases) or if they increase the amount of overlap between  $\kernel(K^X)$ and $\mathcal P_{s\bot}$ (such as curve-wise centering). 

Jointly, these provisions seem to be sufficient to diagnose and safeguard against most serious artifacts of non-identifiability in practice. Our results  indicate that  in many cases,  penalization  allows the reasonable estimation of  coefficient surfaces that are not identifiable in the theoretical model under an additional smoothness assumption, avoiding instead the common assumption that the  estimated coefficient surface lies in the span of the first few eigenfunctions of the covariance operator of the covariate. 

This work drives home the point that we cannot hope to reliably estimate arbitrarily complex effect shapes from functional covariates with low information content. In that sense, assuming smoothness of the coefficient surface and constraining its non-identifiable components to be zero is simply following a principle of parsimony. 
At the same time, substantial interpretation of coefficient surface estimates derived from rank-deficient designs is difficult and has the potential to be very misleading.
Functional principal component regression approaches do not suffer from the potential identifiability issues discussed here, but they do so at the price of restricting the estimated coefficient surface to the span of the estimated functional principal components.
These are significant challenges for the maturing field of functional regression methods, at least in applications where these methods are used not only for prediction, but also for inferring and understanding the underlying data generating processes. It is our hope that the theoretical development and practical examples presented here can serve as a starting point for critical reflection on this important issue.

\section*{Acknowledgements}
The authors wish to thank Ciprian Crainiceanu, Ludwig Fahrmeir,  Jeff
Goldsmith and Clara Happ for helpful discussions and encouragement.
Comments by an anonymous reviewer spurred substantial improvements of our methodology and presentation.
Financial support by the German Research Foundation through the Emmy Noether
Programme, grant GR 3793/1-1 is gratefully acknowledged.

\newpage
\appendix

\section{Proofs}
\subsection{Proof of Proposition \ref{sonjaproposition1}}\label{sec:proof-sonjaproposition1}
Assume that $\mB_t$ is of full rank $K_t$. Then, the design matrix $\mD = \mB_t \otimes (\mX \mW
\mB_s)$ in model \eqref{idmodel} is
rank-deficient if and only if
\begin{enumerate}
 \item[a)] $M < K_s$ or
 \item[b)] if $M \geq K_s$, but $\rg(\mPhi \mW \mB_s) < K_s$ .
\end{enumerate}
\begin{proof}
As $\rg(\mX)=\rg(\mXi \mPhi)=\rg(\mPhi)=M$,
\begin{align*}
\rg(\mPhi \mW \mB_s) &\geq \rg(\mXi \mPhi \mW \mB_s) = \rg(\mX \mW \mB_s) \\
	&\geq \rg(\mXi \mPhi) + \rg(\mPhi \mW \mB_s) - \rg(\mPhi)
	= \rg(\mPhi \mW \mB_s)
\end{align*}
using \citet[][Th.~17.5.1]{Harville1997}.
Thus,
$\rg(\mX \mW \mB_s) = \rg(\mPhi \mW \mB_s)$. The rank will be less than
full if $\rg(\mPhi \mW \mB_s)<K_s$, including if $M<K_s$, as $\rg(\mPhi \mW
\mB_s) \leq M$ by construction.
\end{proof}

\subsection{Proof of Proposition \ref{unique}}\label{sec:proof-unique}
 Let
 $\bP = \lambda_s (\bm{I}_{K_t} \otimes \bP_s) + \lambda_t (\bP_t \otimes \bm{I}_{K_s})$,
 with $\bP_s$ and $\bP_t$ positive semi-definite matrices.
 Assume that $\bm B_t$ is of full rank $K_t$, that
 $\rg(\mP_t) < K_t$ and that $\lambda_s >0, \lambda_t \geq 0$.
 Then, for any $\mf \in \image(\bm D)$ there is a unique minimum
$\min_{ \btheta \in \mathcal{H}_f }\{ \btheta\tr \bP \btheta\}$ if and only if $\kernel(\bm D_s\tr  \bm D_s ) \cap \kernel(\bP_s) = \{\bm 0\}$.

\begin{proof}
We have
 \begin{align*}
 & \min \{\btheta\tr \bP \btheta\} & \quad & \text{s.t.} \quad \btheta \in \mathcal{H}_f \\
=& \min\{ \btheta\tr \bU \bU\tr \bP \bU \bU\tr \btheta\} & \quad & \text{s.t.} \quad \bD \btheta = \mf \\
=& \min\{ (\btheta\tr \bU_+ | \btheta\tr \bU_0) \begin{pmatrix}
 \bU_+\tr \bP \bU_+ & \bU_+\tr \bP \bU_0 \\
 \bU_0\tr \bP \bU_+ & \bU_0\tr \bP \bU_0
 \end{pmatrix}
 \begin{pmatrix}
 \bU_+\tr \btheta \\
 \bU_0\tr \btheta
 \end{pmatrix}\}
 & \quad & \text{s.t.} \quad \bU_+\tr \btheta = \mSigma_+\1 \mV_+\tr \mf.
\end{align*}
 Denote $\mv_+ = \bU_+\tr \btheta$ and $\mv_0 = \bU_0\tr \btheta$, with $(\mv_+\tr, \mv_0\tr)\tr = \bU\tr \btheta$ a bijective re-parametrization of $\btheta$.
Note that
$\mv_+$ is fixed while $\mv_0$ is free to vary within the hyperplane.
Setting the derivative with respect to $\mv_0$ equal to zero yields
\begin{align}\label{eq:min_hf}
\bU_0\tr \bP \bU_0 \mv_0 = - \bU_0\tr \bP \bU_+ \mv_+.
\end{align}
Now, if
$\kernel(\mD_s\tr \mD_s ) \cap \kernel(\bP_s) \neq \{\bm 0\},$
choose $\bm u_s \neq \bm 0$ with
$\mU_{s0} \bm u_s \in \kernel(\mD_s\tr \mD_s ) \cap \kernel(\bP_s)$
and $\bm u_t \neq \bm 0$ with $\mU_t \bm u_t \in \kernel(\bP_t)$. As
$\mU_0 = \mU_t \otimes \mU_{s0}$, we have
\begin{align*}
(\bm u_t\tr \otimes \bm u_s\tr) \bU_0\tr \bP \bU_0 (\bm u_t \otimes \bm u_s) =
\lambda_s (\bm u_t\tr \bm u_t) \bm u_s\tr \mU_{s0}\tr \mP_s \mU_{s0} \bm u_s + \lambda_t \bm u_t\tr \mU_t\tr \mP_t \mU_t \bm u_t (\bm u_s\tr \bm u_s)
= 0.
\end{align*}
Thus, $\bU_0\tr \bP \bU_0$ is not of full rank, no unique solution $\mv_0$ of \eqref{eq:min_hf} exists, so there is no unique minimum $\min_{ \btheta \in \mathcal{H}_f }\{ \btheta\tr \bP \btheta\}$.

On the other hand, if
$\kernel(\mD_s\tr\mD_s ) \cap \kernel(\bP_s) = \{\bm 0\}$,
for any $\bm x$ with $\bm x\tr \mU_{s0}\tr \mP_s \mU_{s0} \bm x = 0$ we
have
$\mU_{s0} \bm x \in \kernel(\mD_s\tr\mD_s ) \cap \kernel(\bP_s) = \{\bm 0\}$
and thus $\bm x = \mU_{s0}\tr \mU_{s0} \bm x = \bm 0$. As this means that
$\mU_{s0}\tr \mP_s \mU_{s0}$ is positive definite,
we also have that
$\bU_0\tr \bP \bU_0 = \lambda_s \mI_{K_t} \otimes (\mU_{s0}\tr \mP_s \mU_{s0}) + \lambda_t (\mU_t\tr \mP_t \mU_t) \otimes \mI_{(K_s - d/K_t)}$ is positive definite and thus invertible. Therefore, there is a unique minimum $\mv_0 = - (\bU_0\tr \bP \bU_0)^{-1} \bU_0\tr \bP \bU_+ \mv_+$
and a unique smoothest point
\bea\label{thetac}
 \btheta_f &=& \bU (\mv_+\tr, - [(\bU_0\tr \bP \bU_0)^{-1} \bU_0\tr \bP \bU_+ \mv_+]\tr)\tr
 = \mH \bU_+ \mv_+ = \mH \btheta,
\eea
with $\mH = (\mI_{K_sK_t} - \bU_0 (\bU_0\tr \bP \bU_0)^{-1} \bU_0\tr\bP )$ 
 and $\btheta_f\tr \bP \btheta_f = \min_{\btheta \in \mathcal{H}_f} \btheta\tr \bP \btheta$. \end{proof}

\subsection{Proof of Proposition \ref{invert}}\label{sec:proof-invert}

Assume that $\bm B_t$ is of full rank $K_t$, that $\rg(\mP_t) < K_t$ and that $\lambda_s >0, \lambda_t \geq 0$. Then, there is a unique penalized least squares solution
for \eqref{penest} if and only if $\kernel(\mD_s\tr
\mD_s ) \cap \kernel(\bP_s) = \{\bm 0\}$.

\begin{proof}
Problem \eqref{penest} has a unique solution iff
$(\bD\tr \bD + \lambda_s (\bm{I}_{K_t} \otimes \bP_s) + \lambda_t (\bP_t \otimes \bm{I}_{K_s})) \geq 0$
is invertible, i.e., positive definite. Now, suppose that
$\kernel(\mD_s\tr \mD_s ) \cap \kernel(\bP_s) = \{\bm 0\}$.
For any $\bm x \in \real^{K_tK_s}$ with
\begin{eqnarray}\nn
\bm x\tr (\bD\tr \bD + \lambda_s (\bm{I}_{K_t} \otimes \bP_s) + \lambda_t (\bP_t \otimes \bm{I}_{K_s})) \bm x = 0,
\end{eqnarray}
we have, with $\mb = \mU\tr \bm x = (b_{kj})_{kj \in \{11, 12, \dots, K_tK_s\}}$,
\begin{align}\label{folge1}
\begin{split}
 \bm x\tr \bD\tr \bD \bm x &= \mb\tr \left(\mSigma_t^2 \otimes
\left(\begin{array}{cc}
 \mSigma_{s+}^2 & \bm 0 \\
 \bm 0 & \bm 0
 \end{array}\right)
 \right)\mb = 0  \\
 \Ra  \quad b_{kj} &= 0 \;\forall\; 1 \leq j \leq d/K_t;\; 1 \leq k \leq K_t,
 \end{split}
\end{align}
and also
\begin{eqnarray}\nn
 \mb\tr [\mI_{K_t} \otimes ((\mU_{s+} | \mU_{s0})\tr \bP_s (\mU_{s+} | \mU_{s0}) )] \mb
\stackrel{\eqref{folge1}}{=} \tilde{\mb}\tr [\mI_{K_t} \otimes ( \mU_{s0}\tr \bP_s \mU_{s0}) ] \tilde{\mb}
 = 0,
\end{eqnarray}
where $\tilde{\mb}$ is obtained by removing the zero entries given by \eqref{folge1} which correspond to $\mU_{s+}$ from $\mb$. Thus, for all $1 \leq k \leq K_t$,
and letting $\tilde{\mb}_k = (b_{k(d/K_t+1)}, \dots, b_{kK_s})$, we have
$\mU_{s0} \tilde{\mb}_k \in \kernel(\mD_s\tr \mD_s ) \cap \kernel(\bP_s) = \{\bm 0\}$. Thus, $\tilde{\mb} = \bm 0$, $\mb = \bm 0$, $\bm x = \bm 0$ and $(\bD\tr \bD + \lambda_s (\bm{I}_{K_t} \otimes \bP_s) + \lambda_t (\bP_t \otimes \bm{I}_{K_s}))$ is of full rank.

On the other hand, if
$\kernel(\mD_s\tr\mD_s ) \cap \kernel(\bP_s) \neq \{\bm 0\}$, there is a
$\bm u_s \neq \bm 0$ with $\mD_s \bm u_s = \bm 0$ and $\bP_s \bm u_s = \bm 0$.
Choose $\bm 0 \neq \bm u_t \in \kernel{\bP_t}$. Then,
\begin{eqnarray}\nonumber
(\bm u_t\tr \otimes \bm u_s\tr) (\bD\tr \bD + \lambda_s (\bm{I}_{K_t} \otimes \bP_s) + \lambda_t (\bP_t \otimes \bm{I}_{K_s})) (\bm u_t \otimes \bm u_s) \\\nn
= \bm u_t\tr \bm B_t\tr \bm B_t \bm u_t \cdot 0 + \lambda_s \bm u_t\tr \bm u_t \cdot 0 + \lambda_t \cdot 0 \cdot \bm u_s\tr \bm u_s = 0.
\end{eqnarray}
Thus,
$(\bD\tr \bD + \lambda_s (\bm{I}_{K_t} \otimes \bP_s) + \lambda_t (\bP_t \otimes \bm{I}_{K_s}))$
is singular and not invertible.
\end{proof}

\end{document}